\documentclass[onecollarge]{svjour2}  
\smartqed  
\usepackage{graphicx,url,cite,amsmath,amssymb,amsfonts,color}

\newenvironment{myquote}{\list{}{\leftmargin=1.25in\rightmargin=0in}\item[]}{\endlist}

\newcommand{\rmd}{{\rm d}}
\newcommand{\bA}{{\bold A}}

\newcommand{\RR}{{\mathbb{R}}}

\newcommand{\cM}{{\mathcal M}}
\newcommand{\cS}{{\mathcal S}}
\newcommand{\cF}{{\mathcal F}}
\newcommand{\cT}{{\mathcal T}}

\newcommand{\bx}{{\bold x}}
\newcommand{\by}{{\bold y}}
\newcommand{\bW}{{\mathbf{W}}}
\newcommand{\bP}{{\mathcal{P}}}

\newcommand{\ba}{{\bold a}}
\newcommand{\bu}{{\bold u}}

\newcommand{\bv}{{\bold v}}
\newcommand{\tbv}{\tilde{\bv}}

\newcommand{\txi}{{\tilde{\xi}}}
\newcommand{\hd}{\hat{\rmd}}
\newcommand{\grad}{{\mbox{\boldmath $\nabla$}}}
\newcommand{\btimes}{{\mbox{\boldmath $\times$}}}
\newcommand{\bzed}{{\mbox{\boldmath $0$}}}
\newcommand{\bxi}{{\mbox{\boldmath $\xi$}}}
\newcommand{\tbxi}{\tilde{\bxi}}

\newcommand{\be}{\begin{equation}}
\newcommand{\ee}{\end{equation}}
\newcommand{\lb}{\label}

\begin{document}

 \title{Spontaneous Stochasticity and Anomalous Dissipation for Burgers Equation} 
\author{Gregory L. Eyink \and Theodore D. Drivas}

\institute{Gregory L. Eyink \at
              Department of Applied Mathematics \& Statistics\\  Department of Physics \& Astronomy\\ The Johns Hopkins University \\
              \email{eyink@jhu.edu} { Tel: 410-516-7201} { Fax: 410-516-7459}      
           \and
           Theodore D. Drivas  \at
              Department of Applied Mathematics \& Statistics\\ The Johns Hopkins University
}

\date{Received: date / Accepted: date}

\maketitle

\begin{abstract}
\noindent 
We develop a Lagrangian approach to conservation-law anomalies in weak solutions of inviscid Burgers 
equation, motivated by previous work on the Kraichnan model of turbulent scalar advection. We show that 
the entropy solutions of Burgers possess Markov stochastic processes of (generalized) Lagrangian trajectories backward 
in time for which the Burgers velocity is a backward martingale. This property is shown to guarantee 
dissipativity of conservation-law anomalies for general convex functions of the velocity. The backward stochastic 
Burgers flows with these properties are not unique, however. We construct infinitely many such stochastic flows, both by 
a geometric construction and by the zero-noise limit of the Constantin-Iyer stochastic representation of viscous 
Burgers solutions. The latter proof yields the spontaneous stochasticity of Lagrangian trajectories backward 
in time for Burgers, at unit Prandtl number. It is conjectured that existence of a backward stochastic flow 
with the velocity as martingale is an admissibility condition which selects the unique entropy solution for Burgers. 
We also study linear transport of passive densities and scalars by inviscid Burgers flows. We show that shock 
solutions of Burgers exhibit spontaneous stochasticity backward in time for all finite Prandtl numbers, 
implying conservation-law anomalies for linear transport. We discuss the relation of our results for Burgers with  
incompressible Navier-Stokes turbulence, especially Lagrangian admissibility conditions for Euler solutions
and the relation between turbulent cascade directions and time-asymmetry of Lagrangian stochasticity.  
\end{abstract}

\section{Introduction}\lb{sec:intro}

\begin{myquote}
{\small {\it There seems to be a strong relation between the behavior of the Lagrangian trajectories and the basic 
hydrodynamic properties of developed turbulent flows: we expect the appearance of non-unique trajectories 
for $Re\rightarrow\infty$ to be responsible for the dissipative anomaly, the direct energy cascade, the dissipation 
of higher conserved quantities and the pertinence of weak solutions of hydrodynamical equations at 
$Re=\infty$.} \\ 

\hfill  --- K. Gaw\c{e}dzki \& M. Vergassola (2000) }\\
\end{myquote}

Energy dissipation in incompressible Navier-Stokes turbulence is, within experimental errors, independent 
of viscosity at sufficiently high Reynolds numbers. For a recent review of the evidence, see
\cite{Pearsonetal02, Kanedaetal03}. This empirical observation motivated  Onsager in 1949 to conjecture that incompressible 
fluid turbulence is described by singular (distributional) Euler solutions that dissipate energy by a nonlinear cascade 
mechanism \cite{Onsager49, EyinkSreenivasan06}. While this conjecture is consistent with all present 
available evidence, quite deep theoretical problems remain. Physically, the Lagrangian interpretation 
of the Eulerian energy cascade is usually in terms of Taylor's vortex-stretching picture \cite{TaylorGreen37,Taylor38}. 
However, Taylor's ideas depend on the validity of the Kelvin circulation theorem, which is very unlikely 
to hold  in the conventional sense for high-Reynolds-number turbulent fluids \cite{Luthietal05,Gualaetal05, Eyink06}. 
Mathematically,  Onsager's conjectured Euler solutions have not yet been obtained as zero-viscosity limits of 
Navier-Stokes solutions. While weak Euler solutions have been constructed which dissipate kinetic energy and 
have the spatial H\"older regularity of observed turbulent fields \cite{Buckmaster13}, such solutions are wildly 
non-unique. Admissibility conditions for weak Euler solutions based only on energy dissipation do not select 
unique  solutions \cite{deLellisSzekelyhidi10}.  

These problems have been resolved, on the other hand, in a toy turbulence model, the  Kraichnan model 
of passive scalar advection by a Gaussian random velocity field which is white-noise in time and 
rough (only H\"older continuous) in space \cite{Kraichnan68, Falkovichetal01}. In this model 
there is anomalous dissipation of the scalar energy due to a turbulent cascade process. In Lagrangian 
terms the turbulent dissipation is explained by  ``spontaneous stochasticity'' of the fluid particle trajectories, 
associated to Richardson explosive dispersion \cite{Bernardetal98}. As quoted in the epigraph 
of this introduction, Gaw\c{e}dzki \& Vergassola \cite{GawedzkiVergassola00} suggested that this non-uniqueness 
and intrinsic stochasticity of Lagrangian trajectories should underlie also the anomalous dissipation in weak 
solutions of hydrodynamic equations relevant to actual fluid turbulence.  Weak solutions of the passive advection 
equation in the Kraichnan model have been rigorously constructed and shown to coincide with solutions
obtained by smoothing the velocity or adding scalar diffusivity and then removing these regularizations 
\cite{EvandenEijnden00, EvandenEijnden01,LeJanRaimond02, LeJanRaimond04}. One can characterize these weak solutions 
by the property that the scalar values are backward martingales for Markov random processes of 
Lagrangian trajectories. 

This successful theory for the Kraichnan model motivated one of us to conjecture a similar ``martingale 
hypothesis'' for fluid circulations in the weak solutions of incompressible Euler equations that are 
believed to be relevant for turbulence \cite{Eyink06}. For smooth solutions of Euler equations, the backward martingale  
property reduces  to the usual Kelvin Theorem on conservation of circulations. However, for singular solutions
it imposes an ``arrow of time'' which was proposed as an infinite set of admissibility conditions to select
``entropy'' solutions of the Euler solutions \cite{Eyink07, Eyink10}. This conjecture assumes spontaneous stochasticity 
in high-Reynolds Navier-Stokes turbulence, for which numerical evidence has been obtained in 
studies of 2-particle dispersion \cite{Sawfordetal08, Eyink11, Bitaneetal13}. Subsequently, 
in very beautiful work,  Constantin \& Iyer \cite{ConstantinIyer08} established a characterization of the 
solutions of the incompressible Navier-Stokes solutions as those space-time velocity fields for which the 
fluid circulations are backward martingales of a stochastic advection-diffusion process (see also \cite{Eyink10}). 
This ``stochastic Kelvin Theorem'' is the exact analogue for Navier-Stokes of the property proposed earlier for entropy 
solutions of Euler equations. Of course, the Navier-Stokes result does not imply that for Euler and, at this time,
the zero-viscosity limit is so poorly understood that a mathematical proof (or disproof) for Euler does not seem 
to be forthcoming anytime soon.

There are simpler PDE problems, however, where the zero-viscosity limit is much better understood. 
These include scalar conservation laws in one space dimension \cite{Bressan09,Dafermos05}, with the 
Burgers equation \cite{Burgers39, BecKhanin07} as a prominent example.  The scalar conservation laws 
possess weak solutions that are uniquely selected by entropy admissibility conditions and which coincide 
with solutions obtained by the zero-viscosity limit. The Burgers equation, in particular, has long been a 
testing ground for ideas about Navier-Stokes turbulence\footnote{Note, furthermore, that the general 
scalar conservation law in one-dimension $u_t+(f(u))_x=0$ with a strictly convex flux function $f$ is 
equivalent for smooth solutions to Burgers equation for the associated velocity field $v=f'(u).$ This 
equivalence extends to viscosity-regularized equations in a slightly modified form. A simple calculation 
shows that $u_t+(f(u))_x=\varepsilon \ u_{xx}$ is equivalent to $v_t+\left[\frac{1}{2}v^2+\varepsilon g(v) \right]_x =\varepsilon \ v_{xx},$ 
where $g(v)=1/\hat{f}''(v)$ and $\hat{f}(v)$ is the Legendre dual of the convex function $f(u).$ Hence, an entropy 
solution $u$ of $u_t+(f(u))_x=0$ should give an entropy solution $v=f'(u)$ of Burgers, and inversely.}. It is 
therefore a natural question whether the known entropy solutions of inviscid Burgers satisfy a version 
of the martingale property conjectured for ``entropy solutions'' of incompressible Euler. Since smooth 
solutions of inviscid Burgers preserve velocities along straight-line characteristics, the natural 
conjecture for Burgers is that the Lagrangian velocity is a backward martingale. As a matter of fact,
Constantin and Iyer \cite{ConstantinIyer08} established exactly such a characterization of the 
solutions of the viscous Burgers equation. In order for such a representation to hold also for the 
zero-viscosity limit, there must be a form of ``spontaneous stochasticity'' for Burgers flows. It has been 
argued that these flows are only coalescing and that stochastic splitting is absent \cite{BauerBernard99}. 
This is true,however, only forward in time. The natural martingale property involves instead flows backward
in time and it is plausible that there should exist a suitable stochastic inverse of the forward coalescing flow.  

A main result of this paper is that there are indeed well-defined Markov inverses of the forward 
coalescing flows for the entropy solutions of inviscid Burgers, such that the Burgers velocity is a 
backward martingale of these stochastic processes. This result implies a stochastic representation 
of the standard entropy solutions of inviscid Burgers exactly analogous to the Constantin-Iyer (C-I) 
representation of viscous Burgers solutions. Interestingly, there is more than one way to construct 
such a stochastic inverse (in contrast to the Kraichnan model, where the stochastic process of 
backward Lagrangian trajectories appears to be essentially unique \cite{LeJanRaimond02,LeJanRaimond04})
\footnote{The ``essential uniqueness'' is that of the stochastic backward process for a 
given weak solution of the passive-scalar advection equation. The Kraichnan model for an intermediate 
regime of compressibility has  distinct weak solutions in the simultaneous limit $\nu,\kappa\rightarrow 0,$ obtained 
by holding fixed different values of the ``turbulent Prandtl number'' \cite{EvandenEijnden00, EvandenEijnden01}.
The backward stochastic process is uniquely fixed by that limit, however, which fully specifies the boundary 
conditions at zero-separation. We shall see that the case is otherwise with Burgers, which has infinitely many 
distinct stochastic inverse flows for the same, unique dissipative weak solution.}. 
We obtain one set of stochastic inverses by a direct geometric construction, closely related to recent work of  
Moutsinga \cite{Moutsinga12}. We obtain another stochastic inverse by the zero-viscosity limit of the backward 
diffusion processes in the Constantin-Iyer representation, demonstrating spontaneous stochasticity for Burgers flows   
backward in time at unit Prandtl number\footnote{``Prandtl number'' here is the ratio of the viscosity 
to the square-amplitude of a white-noise term in the Lagrangian particle equation. This is exactly the 
standard Prandtl (or Schmidt) number for passive scalar advection, when scalar diffusivity is represented 
by stochastic particle motion. See section \ref{sec:scalar}.}. The stochastic inverse flows we obtain are (backward) Markov 
jump-drift processes supported on generalized solutions of the Lagrangian particle equations of motion 
(generalized characteristics in the sense of Dafermos \cite{Dafermos05}.)  Although not themselves 
unique, each constructed backward stochastic flow enjoys the properties discussed above 
and provides a representation of the unique entropy  solutions of Burgers.  Furthermore, we 
show that the backward martingale property of the Burgers velocity is exactly what is required 
to make the solutions dissipate convex entropies. For this purpose, we derive a novel 
Lagrangian formula for inviscid Burgers dissipation.  We conjecture that existence of a stochastic 
process of generalized characteristics with the backward martingale property for velocities 
is an admissibility condition for inviscid Burgers which uniquely selects the standard entropy solution.  

A second main contribution of this work is to the study of linear transport by Burgers 
flows, see Woyczy\'nski \cite{Woyczynski98} and Bauer \& Bernard \cite{BauerBernard99}. 
Because Burgers provides a mathematically tractable example of compressible turbulence with shocks, 
it is possible to gain insight into physically more relevant transport problems for scalars and densities 
in compressible Navier-Stokes and magnetohydrodynamic turbulence. A key question here 
also is the existence or not of anomalies in the conservation laws of scalars and densities 
associated to their own turbulent cascades. The lesson of the Kraichnan model is that this 
question is directly related to the spontaneous stochasticity of the turbulent flow for general 
Prandtl numbers \cite{Falkovichetal01}. Bauer \& Bernard concluded on this basis that there 
are no dissipative anomalies for scalars and densities transported by Burgers because they argued 
that there is no spontaneous stochasticity for Burgers flows \cite{BauerBernard99}. 
For scalar dissipation, however, it is backward-in-time stochasticity which is relevant and we have 
shown that this property holds for Burgers, at least at Prandtl number unity. Extending that result,
we further show that spontaneous stochasticity holds backward in time in some shock solutions of 
Burgers for all finite values of the Prandtl number. These results imply the existence of conservation-law 
anomalies  both for densities and for scalars advected by Burgers\footnote{The Bauer-Bernard picture 
with vanishing anomalies may be valid for the infinite Prandtl number limit, which is more subtle and only briefly 
discussed in this paper. See the recent work \cite{FrishmanFalkovich14}.}. They also lead to a new notion of a 
``Lagrangian weak solution'' for passive scalars in a compressible flow, when the standard
notion of distributional weak solution in the Eulerian formulation is not available.  
Finally, we discuss the importance of the time-irreversibility of Burgers equation and 
the associated differences with the time-reversible Kraichnan model. We suggest that 
the direction of turbulent cascades is related generally in irreversible fluid models 
to the time-asymmetry of Lagrangian particle behavior.  

The detailed contents of this paper are as follows: In section 2 we derive our Lagrangian
formula for conservation-law anomalies in Burgers. Section 3 presents the geometric 
construction of the stochastic inverse to the forward coalescing flow for Burgers. In section 4
we study the zero-viscosity limit of the Constantin-Iyer representation of viscous Burgers 
solutions and establish spontaneous stochasticity backward in time for unit Prandtl number. 
Section 5 discusses the non-uniqueness of the backward stochastic flows for Burgers but 
their conjectured unique characterization of the entropy weak solution.  Section 6 discusses 
passive densities in Burgers and sticky-particle dynamics. Section 7 studies passive scalars,
their conservation-law anomalies, and establishes spontaneous stochastic backward in time
for Burgers shocks at arbitrary Prandtl numbers. Section 8 discusses the time-asymmetry of 
Lagrangian particle statistics and its possible relation to dissipative anomalies. Two appendices 
provide more technical details for some of the proofs. 

\section{Lagrangian Formulation of Anomalous Dissipation}\lb{sec:lagdiss}

We derive first a Lagrangian expression for dissipative anomalies of inviscid Burgers. 
 
\subsection{Basic Burgers Facts}\lb{sec:basic} 

Before beginning, we remind the reader of some standard results about Burgers, many quite elementary. 
For example, see \cite{BecKhanin07}.  Let $u$ be a smooth solution of the inviscid Burgers equation 
for initial data $u_0$ at time $t_0.$ Using the standard method of characteristics, one can see that
\begin{align}
 x &= a+ (t-t_0) u_0(a), \,\,\,\,
u(x,t)=u_0(a). \lb{char} \end{align}
Note that 
$$ \xi_{t_0,t}(a) = a + (t-t_0)u_0(a) $$
is the Lagrangian flow map of fluid mechanics, with inverse $\alpha_{t_0,t}=\xi_{t_0,t}^{-1}$ the 
``back-to-labels''' map so that $u(x,t)=u_0(\alpha_{t_0,t}(x)).$ All of the following are simple consequences of (\ref{char}):
$$  u'(x,t) =u_0'(\alpha_{t_0,t}(x)) \alpha_{t_0,t}'(x) $$ 
$$ \xi_{t_0,t}'(a)  = {1}+ (t-t_0) u_0'(a) $$
$$ \alpha_{t_0,t}'(x) = {1}- (t-t_0) u'(x,t) = [\xi_{t_0,t}'(a)]^{-1}$$
and thus
$$
u'(x,t) = \frac{u_0'(a)}{1+(t-t_0) u_0'(a)}.
$$
It follows from the latter formula that, wherever $u'(a)<0$ at any initial point $a$, a shock will form in finite time 
from smooth initial data $u_0(a).$ The singularity will occur (unless the particle is absorbed first by another shock) at time 
$$
t= t_0+ \frac{1}{\max\{0,-u_0'(a)\}}.
$$  
The first shock occurs at the minimum of the above quantity, related to the maximum of the negative velocity gradient. 
At later times, all of the previous results for smooth solutions are valid at points between shocks.  

We consider Burgers solutions of bounded variation with countably many shocks located at coordinates 
$\{x_i^*\}_{i=1}^\infty$ at time $t$. Let $u^-_i$ be the velocity immediately to the left of the $i$th shock 
and $u^+_i$ the velocity immediately to the right. The \emph{Rankine-Hugoniot jump conditions} require that 
the shock velocity $u_i^*=dx_i^*/dt$ for any weak solution be an average:
\begin{align}\label{shock speed}
u_i^*=\frac{u^-_i+ u^+_i}{2}.  
\end{align}
{\it Entropy solutions} of inviscid Burgers have the property $u^-_i >u^+_i.$ As a matter of fact, it is 
well-known that the energy conservation anomaly at a Burgers shock is $\frac{1}{12}(u^+_i-u^-_i)^3,$
which is negative (dissipative) precisely when $u^-_i >u^+_i.$ This is also the Lax admissibility condition 
for weak solutions \cite{Lax57} in the context of Burgers. Thus, each shock corresponds to a 
Lagrangian interval $[a_i^-,a_i^+]$  such that $u_i^\pm = u_0(a_i^\pm)$ and 
\begin{align}
 x_i^* &= a^-_i+ t u^-_i = a^+_i+ t  u^+_i. \label{endpts}
\end{align}
The union of shock intervals in the Lagrangian space is denoted below as 
$
S=\bigcup_{i=1}^\infty[a_{i}^-,a_{i}^+].
$

\subsection{Dissipative Anomalies}\lb{sec:dissanom} 

Our goal in this section is to derive fundamentally Lagrangian expressions for dissipative anomalies
in inviscid Burgers, analogous to those obtained for integral invariants of passive scalars in the Kraichnan 
model \cite{Bernardetal98,GawedzkiVergassola00} . 
Thus let $\psi$ be a continuous function and $\Psi$ its anti-derivative.  Take $t_0=0$ for simplicity. Then 
\begin{align*}
\int_\RR\rmd x \ \psi(u(x,t))&=\int_{\RR\setminus\{x_i^*\}_{i=1}^\infty}\rmd x \ \psi(u(x,t))
=\int_{\RR\setminus\{x_i^*\}_{i=1}^\infty} \rmd x \ \psi(u_0(\alpha_{t_0,t}(x))) \\
&=\int_{\RR\setminus S} \rmd a \ \psi(u_0(a)) \ \xi_{t_0,t}'(a)\\
&=\int_{\RR} \rmd a \ \psi(u_0(a))\big( 1+ t u_0'(a)\big) - \int_{S} \rmd a \ \psi(u_0(a))\big( 1+ t u_0'(a)\big)\\
&=\int_{\RR} \rmd a \ \psi(u_0(a)) + t\int_{\RR} \rmd a \ \frac{d}{da}\Psi(u_0(a))- \int_{S} \rmd a \ \psi(u_0(a))\big( 1+ t u_0'(a)\big)\\
&=\int_{\RR} \rmd a \ \psi(u_0(a)) - \int_{S} \rmd a \ \psi(u_0(a))\big( 1+ t u_0'(a)\big)
\end{align*}
We used the assumption $\lim_{a\rightarrow\pm\infty} u_0(a)=u_\infty$ to set $\int_{\RR} \rmd a \ \frac{d}{da}\Psi(u_0(a))=0.$
We see that $\int_\RR\rmd x \ \psi(u(x,t))$ is conserved for a smooth Burgers solution, when $S=\emptyset.$

We now consider the case of weak solutions with shocks. We can rewrite the second term:
\begin{align*}
 \int_{S} \rmd a \ \psi(u_0(a))\big( 1+ t u_0'(a)\big)&= \int_{S} \rmd a \ \psi(u_0(a))+ t\int_{S} \rmd a \ \psi(u_0(a)) u_0'(a)\\
    &= \sum_{i=1}^\infty \bigg[ \int_{a_i^-}^{a_i^+} \rmd a \ \psi(u_0(a))+ t \int_{a_i^-}^{a_i^+} \rmd a \ \psi(u_0(a)) u_0'(a)\ \bigg]\\
    &= \sum_{i=1}^\infty \bigg[ \int_{a_i^-}^{a_i^+} \rmd a \ \psi(u_0(a))- t \int_{u_i^+}^{u_i^-} \rmd u \ \psi(u) \bigg]
\end{align*}
Thus,
\be
\int_\RR\rmd x \ \psi(u(x,t)) - \int_{\RR} \rmd a \ \psi(u_0(a)) 
= -\sum_{i=1}^\infty \bigg[ \int_{a_i^-}^{a_i^+} \rmd a \ \psi(u_0(a))- t \int_{u_i^+}^{u_i^-} \rmd u \ \psi(u) \bigg]
\lb{lagdiss} \ee
The right-hand side is a Lagrangian representation of the conservation law anomaly. 

A Burgers solution $u$ is a {\it dissipative} if, for any convex function $\psi,$ 
\be
\int_{a_i^-}^{a_i^+} \rmd a \ \psi(u_0(a))\geq  t \int_{u^+_i}^{u_i^-} \rmd u \ \psi(u), \,\,\,\,i=1,2,\dots
\lb{diss} \ee
Dividing by $a_i^+-a_i^-$ and using the relationship \eqref{endpts}, this is equivalent to 
\be
\frac{1}{a_i^+-a_i^-}\int_{a_i^-}^{a_i^+} \rmd a \ \psi(u_0(a))\geq  \frac{1}{u_i^- -u_i^+} \int_{u_i^+}^{u_i^-} \rmd u \ \psi(u),
\,\,\,\,i=1,2,\dots
\lb{diss2} \ee
Since both $\psi(u)=u$ and $\psi(u)=-u$ are convex functions, any dissipative solution must satisfy the relation 
\be {{1}\over{a^+_i-a^-_i}}\int_{a^-_i}^{a^+_i} u_0(a)\, da  = {{1}\over{2}}(u^-_i+u^+_i),  \lb{avrgcond} \ee
which will prove fundamental to our later work. 
Note 
that (\ref{avrgcond}) is equivalent to the standard ``Maxwell construction'' of the dissipative solution at shocks, in which 
one chooses the Lagrangian map of the weak solution to satisfy  $\xi_{t_0,t}^* (a)= x_i^*(t)$ for $a\in [a_i^-,a_i^+],$ 
under the constraint 
\be \int_{a^-}^{a^+} da\, \left[\xi_{t_0,t}(a) - \xi_{t_0,t}^*(a)\right]=0, \lb{max} \ee
with $\xi_{t_0,t}(a) = a +u_0(a) t $ the naive Lagrangian map \cite{BecKhanin07}. 
To see this, substitute the definitions of the maps and integrate to 
give an equivalent expression of the Maxwell construction as 
$$ x_i^*(t) = \frac{1}{2}(a_i^-+a_i^+) + \frac{t}{a_i^+-a_i^-}\int_{a_i^-}^{a_i^+} u_0(a) \ \rmd a . $$
On the other hand, the average of the two expressions in (\ref{endpts}) gives 
\be x_i^*(t) = \frac{1}{2}(a_i^-+a_i^+) + \frac{t}{2}(u_i^-+u_i^+), \lb{avrgcond2} \ee
from which (\ref{avrgcond}) is obviously equivalent to (\ref{max}). 

We now show that (\ref{avrgcond2}) with $u_i^->u_i^+$ implies (\ref{diss}), at any 
final time $t_f.$ Since the argument applies to every shock, we hereafter drop the $i$ subscript. 
The argument is best understood graphically, so we refer to the Fig.1 below which plots a typical Burgers shock: 

\begin{figure}[!ht]
\begin{center}
\includegraphics[height=3in,width=3.5in]{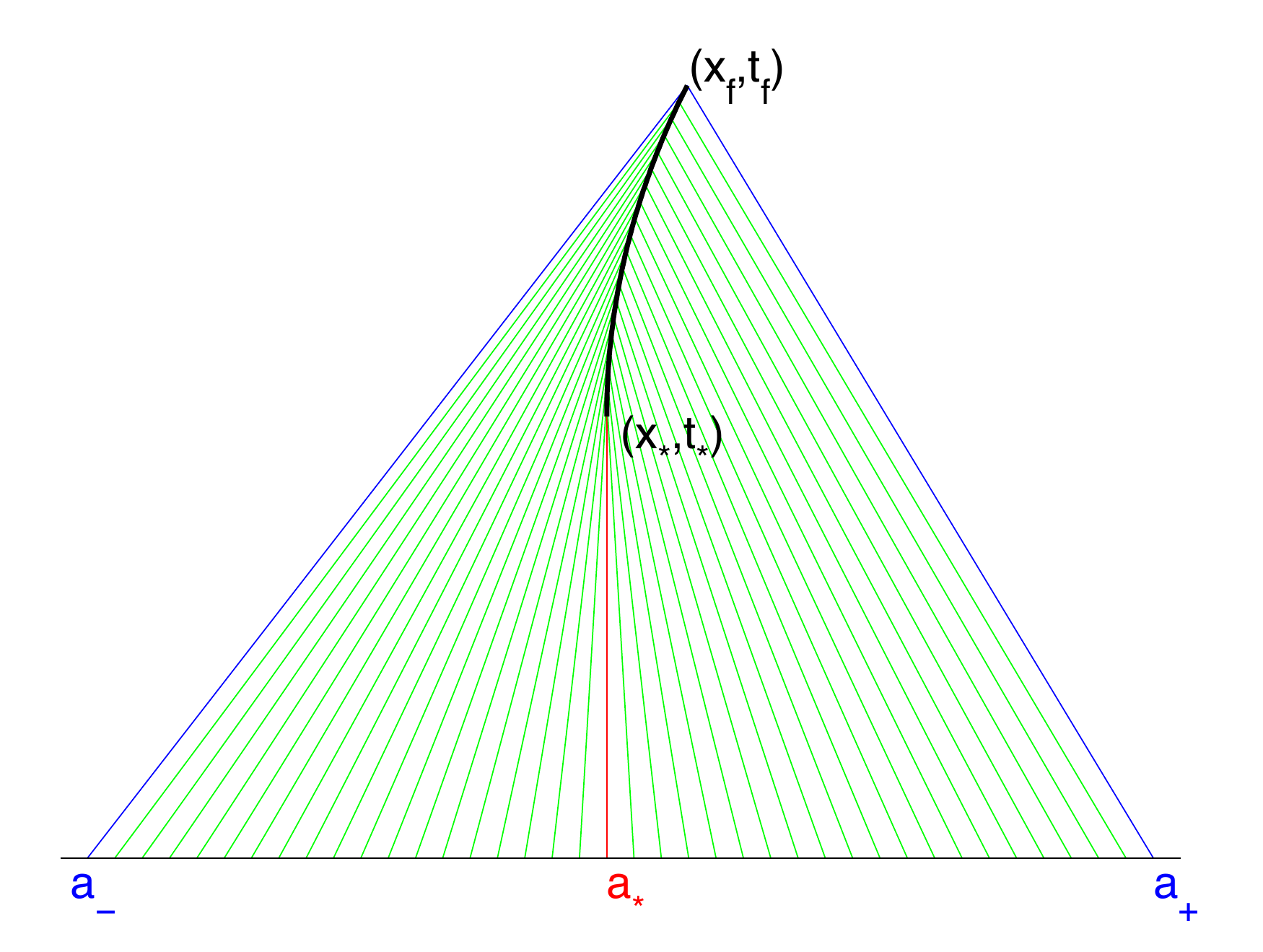}
\caption{{\small {\it Spacetime Plot of a Burgers Shock.} Shown in \textcolor{green}{green} are 
the straight lines corresponding to the smooth particle motions. These converge onto the shock curve
in {\bf black}, which begins at $(x_*,t_*)$ and ends at $(x_f,t_f),$ the final time considered. On the abscissa 
is the space of Lagrangian positions at time $0,$ showing the shock interval $[a_-,a_+]$ and, in 
\textcolor{red}{red}, the label $a_*$ and straight-line characteristic where the shock first forms at time $t_*.$   
 }}
\end{center}
\end{figure}

\noindent Note that the straight characteristic passing through the initial point with label $a$ has slope 
equal to $1/u_0(a).$ Thus, this graph represents the configuration used to obtain the average on the 
lefthand side of (\ref{diss2}). On the other hand, the righthand side of (\ref{diss2}) is obtained from a 
uniformized configuration in which the true initial velocity $u_0(a)$ at each point $a$ is replaced by an 
``apparent initial velocity'' $(x^*_f-a)/t_f.$ This configuration is represented in Fig.2 below by the straight 
line drawn from each point $(a,0)$ to the final point $(x_f^*,t_f).$ The inequality in (\ref{diss2}) is the 
statement that the uniform distribution on the velocity interval $[u_+,u_-]$ is less spread out 
than the distribution of the true initial velocity, as measured by the convex function $\psi.$ To show this,
we can gradually ``lift'' the characteristic lines along the shock curve $x_*(s)$ from $s=t_*$ to $s=t_f$. 
We can expect that the integral is successively decreased by this operation. 
To make this argument analytically, we introduce the function 
$$\Delta_\psi(s) = \int_{a_-(s)}^{a_+(s)} \psi\left(\frac{x_*(s)-a}{s}\right) \ \rmd a 
    + \int_{[a_-,a_-(s)]\cup [a_+(s),a_+]} \psi(u_0(a)) \ \rmd a, $$

\begin{figure}[!ht]
\begin{center}
\includegraphics[height=3in,width=3.5in]{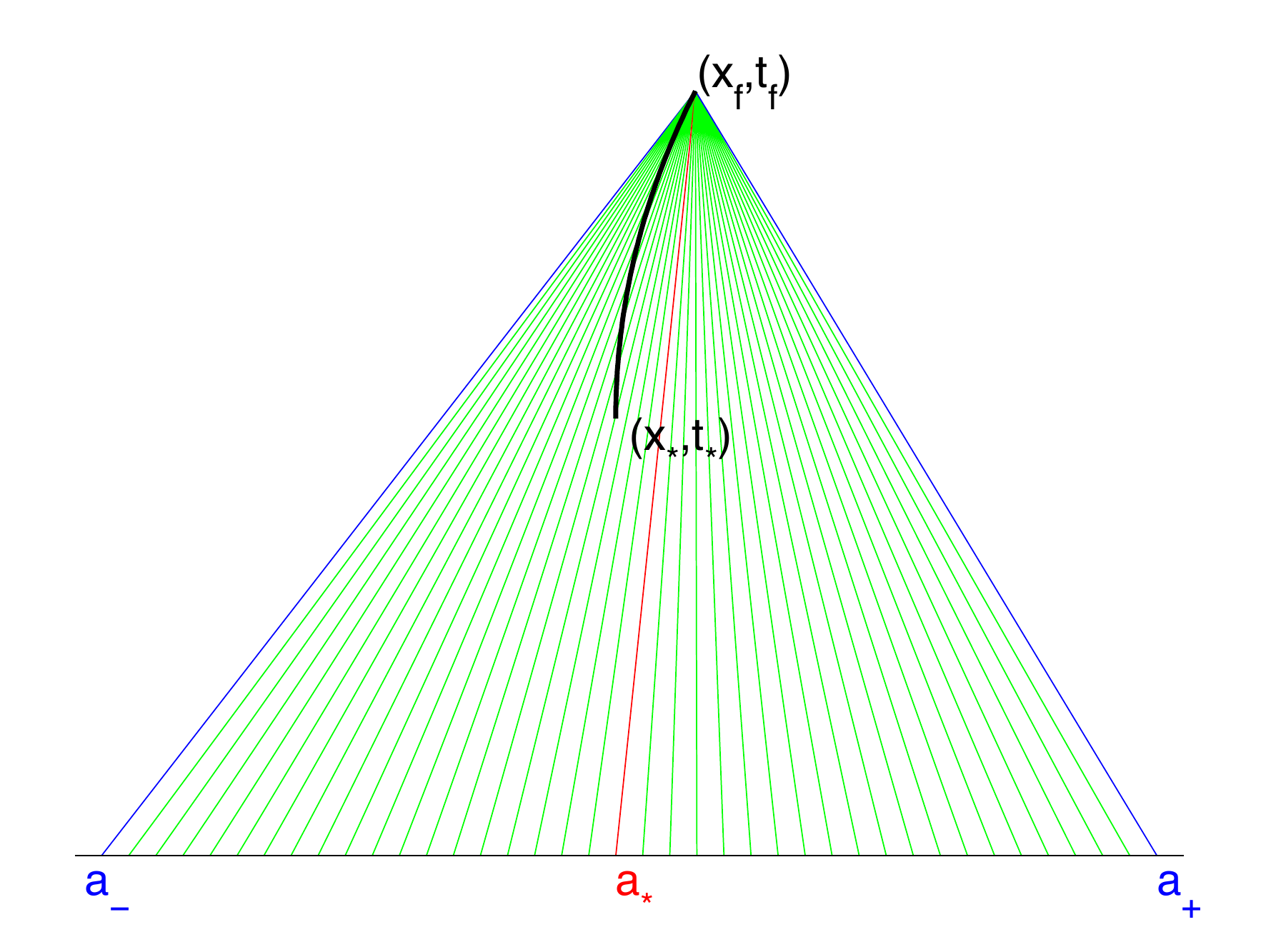}
\caption{{\small {\it ``Uniformized'' Burgers Shock.} Compared with the previous Fig.1, 
all straight-line characteristics have been replaced by straight lines from initial
point $(a,0)$ to the final point $(x_f^*,t_f).$}}
\end{center}
\end{figure}

\noindent 
for $s\in [0,t],$ where $[a_-(s),a_+(s)]$ is the Lagrangian interval at time 0 for the shock located at $x_*(s)$ 
at time $s.$ Note that for $s<t_*,$ the time of first appearance of the shock, 
$$ \Delta_\psi(s) = \int_{a_-}^{a^+} \psi(u_0(a)) \ \rmd a, $$      
while for $s=t$
$$ \Delta_\psi(t) = \int_{a_-}^{a_+} \psi\left(\frac{x_*(t)-a}{t}\right) \ \rmd a = t \int_{u^+}^{u^-} \psi(u) \ \rmd u.  $$
Thus, the total dissipative anomaly over time interval $[0,t]$ (for a single shock) is the difference $\Delta_\psi(t)-\Delta_\psi(0).$
We shall show that $\Delta_\psi(s)$ is non-increasing in $s.$ Taking the $s$-derivative and using (\ref{endpts}) gives
$$ \frac{d}{ds}\Delta_\psi(s) = \frac{1}{s} \int_{a_-(s)}^{a_+(s)} \psi'\left(\frac{x_*(s)-a}{s}\right) \left(u_*(s)- \frac{x_*(s)-a}{s}\right) \ \rmd a $$
Convexity of $\psi$ implies that $\psi\left(\frac{x_*(s)-a}{s}\right) + \psi'\left(\frac{x_*(s)-a}{s}\right) \left(u_*(s)- \frac{x_*(s)-a}{s}\right)
\leq \psi(u_*(s))$ and thus
\be  \frac{d}{ds}\Delta_\psi(s) \leq  \frac{1}{s} \int_{a_-(s)}^{a_+(s)} \left[\psi(u_*(s))-\psi\left(\frac{x_*(s)-a}{s}\right)\right] \ \rmd a. \lb{Psip} \ee
On the other hand, condition (\ref{avrgcond2}) for time $s$ can be rewritten as 
\be u_*(s) = \frac{1}{a_+(s)-a_-(s)}\int_{a_-(s)}^{a_+(s)} \frac{x_*(s)-a}{s} \ \rmd a, \lb{great} \ee
so that convexity of $\psi$ yields also
\be \psi(u_*(s)) \leq \frac{1}{a_+(s)-a_-(s)}\int_{a_-(s)}^{a_+(s)} \psi\left(\frac{x_*(s)-a}{s}\right) \ \rmd a, \lb{great2} \ee
and hence the non-positivity of the righthand side of (\ref{Psip}).  Thus, $\Delta_\psi$ is non-increasing, and $\Delta_\psi(t)\leq \Delta_\psi(0),$
which is equivalent to (\ref{diss}). 

This proof gives a simple Lagrangian interpretation of the dissipative anomaly 
for Burgers equation: information about the initial velocity is ``erased'' as the particles fall into the shock and 
the initial velocity distribution is replaced by a uniform distribution in the shock interval. This decreases the average 
value of any convex function of the velocities because, instantaneously, the velocities are mixed 
(homogenized) by the shock to be closer to its own velocity $u_*(s).$\footnote{This result is a sort of Burgers-equation 
version of Landauer's principle in the physical theory of computation, which states that erasure of information implies 
entropy production \cite{Landauer61}. One may also see some resemblance with the generalized second 
law of black-hole thermodynamics \cite{Mukohyama97}, with the shock being analogous to the event horizon
of the black hole.
} Note that the above argument yields
a new Lagrangian expression for the dissipative anomaly: 
\begin{eqnarray}
&& \int_\RR\rmd x \ \psi(u(x,t)) - \int_{\RR} \rmd a \ \psi(u_0(a)) \cr
&& \,\,\,\,\,\,\,\,\,\,\,\,\,\,\,\,\,\,\,\,
=\sum_{i=1}^\infty \int_0^t \frac{\rmd s}{s} \int_{a_i^-(s)}^{a_i^+(s)} \psi'\left(\frac{x_i^*(s)-a}{s}\right) \left(u_i^*(s)- \frac{x_i^*(s)-a}{s}\right) \ \rmd a \cr
&& \,\,\,\,\,\,\,\,\,\,\,\,\,\,\,\,\,\,\,\,
=\sum_{i=1}^\infty \int_0^t \frac{\rmd s}{s} \int_{a_i^-(s)}^{a_i^+(s)} \Big(\psi(u^*_i(s))-\psi(u_i(s)) - D_\psi^{u_i(s)}(u^*_i(s),u_i(s))\Big)  \ \rmd a. 
\lb{lagdiss2} \end{eqnarray}
We have introduced here the notation $u_i(s)=\frac{x_i^*(s)-a}{s}$ and used the definition of the {\it Bregman divergence} 
between $u$ and $u^*$ with respect to the convex function $\psi$ \cite{Bregman67}:  
$$ D_\psi^u(u^*,u) = \psi(u^*)-\psi(u) - \psi'(u) \cdot (u^*- u). $$
Instantaneously, one has
\be \frac{d}{dt} \int_\RR\rmd x \ \psi(u(x,t)) = \frac{1}{t}
\sum_{i=1}^\infty \int_{a_i^-(t)}^{a_i^+(t)} \Big(\psi(u^*_i(t))-\psi(u_i(t)) - D_\psi^{u_i}(u_i^*(t),u_i(t))\Big)  \ \rmd a. \lb{lagdiss3} \ee
Since $D_\psi^u(u^*,u)\geq 0,$ we can easily see that the contribution from each shock to the anomaly is non-positive 
using inequality (\ref{great2}). 

We recall the standard Eulerian result for the dissipative anomaly
$$ \frac{d}{dt} \int_\RR\rmd x \ \psi(u(x,t)) = \sum_{i=1}^\infty \Big(u_i^*(t) (\psi(u^-_i)-\psi(u^+_i))-(J(u_i^-)-J(u_i^+))\Big). $$
Here $(\psi,J)$ is a Lax entropy pair with entropy flux function defined by 
$$ J(u) = \int \rmd u \ u\ \psi'(u).$$
See \cite{Dafermos05,Bressan09}.  Our Lagrangian formula makes connection with the work of Khanin \& 
Sobolevski\u{\i} \cite{KhaninSobolevski10} on particle dynamics for Hamilton-Jacobi equations. 
They defined a ``dissipative anomaly'' which measured the rate of the difference in 
the action functional between true action minimizers and trajectories of particles on shocks. For Burgers as a 
Hamilton-Jacobi equation, the Hamiltonian and Lagrangian coincide with the convex function $\psi(u)=\frac{1}{2}u^2.$
With this choice of $\psi,$ the ``dissipative anomaly'' of  \cite{KhaninSobolevski10} is the maximum over $\pm$
of the Bregman divergences $D_\psi^{u_i^\pm}(u_i^*,u_i^\pm)=\frac{1}{2}|u_i^\pm-u_i^*|^2.$ Further relations 
with their work will be explored in section \ref{sec:reversal}. 

However, we first exploit the results of the present section to show how to construct, for any entropy solution $u$ of 
inviscid Burgers, a backward Markov process of generalized solutions of the ODE $dx/dt=u(x,t)$. This process is 
thus a generalized (stochastic) inverse of the forward coalescing flow for inviscid Burgers, which has been considered 
by many authors \cite{BauerBernard99, Bogaevsky04}. The essential property of the stochastic inverse 
considered here is that the velocities of the process are backward martingales, generalizing the result for smooth solutions 
of inviscid Burgers that velocities are Lagrangian invariants (preserved along characteristics). As we shall see, it is this 
backward martingale property which implies Lagrangian dissipativity of the entropy solution and it is natural to conjecture 
that this property uniquely characterizes the entropy solution. On the other hand, the stochastic inverses with the 
above stated properties are not themselves unique.  In the following section \lb{sec:geom} we construct a set of such inverses 
by a direct geometric method. Then in section \ref{sec:zerovisc} we obtain another such stochastic inverse via the zero-viscosity 
limit of the backward diffusion process in the Constantin-Iyer representation of viscous Burgers solutions. 

\section{Geometric Construction}\lb{sec:geom} 


To present the geometric construction first in the simplest case, we consider the situation that a single shock has formed 
at time $t_*>0$ and consider a later time $t_f>t_*,$ but before the shock in question has merged with any other. 
The location of this shock for times $t\in [t_*,t_f]$ is denoted by $x_*(t)$ and, at the final time, $x_*^f=x_*(t_f).$  
The random process will consist of continuous curves $x(t)$ over the time interval $[0,t_f]$ satisfying 
$x(t_f)=x_*^f$ a.s. The guiding idea of the construction is to consider the interval of Lagrangian positions 
$[b_-^f,b_+^f]$ at any time $t_0\in [0,t_*)$ which belong to the shock at time $t_f$ and to assume a uniform 
probability distribution over the positions $b$ in this interval. The definition of the random process 
can be understood geometrically, with reference to Fig.1 for the choice $t_0=0.$ The uniform probability 
on the interval $[a_-,a_+]$ is mapped to the shock curve by the straight-line characteristics. Points $a$ 
to the left of $a_*$  map to the shock curve a probability density $p_-(\tau)$ at the time $\tau=(x_*(\tau)-a)/u(a)$ 
when the characteristic enters the shock. Likewise, points $a$ to the right of $a_*$  map to the shock curve 
a probability density $p_+(\tau)$ at the time also determined by $\tau=(x_*(\tau)-a)/u(a).$ Backward 
in time, the random process corresponds to paths which follow the shock curve until they leave either 
to the right or to the left of the shock at time $\tau$ with probability densities $p_\pm(\tau),$ respectively. After leaving the shock,
the paths of the random process move along straight-line characteristics backward in time to initial time $0.$

We now describe the construction analytically, for any choice of initial time $t_0\in [0,t_*).$
Let $u_\pm(\tau)$ be the velocities to the right/left of the shock at times $\tau\in [t_*,t_f].$ The realizations of the 
random process which we construct are of the form
\be  x_\pm(t;\tau) = \left\{ \begin{array}{ll}
                                       x_*(t) & t\geq \tau \cr
                                       x_*(\tau) + u_\pm(\tau)\cdot (t-\tau) & t< \tau 
                                       \end{array} \right. , \,\,\,\, t\in [0,t_f],\tau \in [t_*,t_f]. \lb{real} \ee
These are generalized solutions of the ODE $dx/dt=u(x,t)$ in the sense of \cite{BauerBernard99}. 
That is, they satisfy
$$ D_t^+x(t)=\bar{u}(x(t),t) $$
with $D_t^+x(t)=\lim_{\epsilon\rightarrow 0+} \frac{x(t+\epsilon)-x(t)}{\epsilon}$ and $\bar{u}(x,t)=\frac{1}{2}(u(x+,t)+u(x-,t)).$
Thus, as stated above, the curve coincides with the shock moving backward to time $\tau$ and then branches 
off to the right/left as a particle trajectory for a smooth solution of Burgers. To specify the process we need only 
give the probability densities $p_\pm(\tau)$ to branch off at time $\tau$ which satisfy 
\be \int_{t_*}^{t_f} \rmd\tau \ [p_+(\tau)+p_-(\tau)] =1. \lb{norm} \ee
To assign these probabilities we note that the positions at time $t$ of particles located at $b$ at time $t_0$
are given by $x=b+u(b,t_0)(t-t_0).$ Hence, these particles hit the shock at time $\tau$ for the two points 
$b_\pm(\tau)$ determined by 
\be x_*(\tau) = b_\pm + u(b_\pm,t_0)(\tau-t_0). \lb{transf} \ee
We now assume a uniform probability distribution of these $b$ on the interval $[b_+^f,b_-^f],$ that is, 
$b_+$ is distributed on $[b_*,b_+^f]$ with density $db/(b_+^f-b_-^f)$ and 
$b_-$ is distributed on $[b_-^f,b_*]$ with density $db/(b_+^f-b_-^f),$ where $b_*$ is the particle location 
which shocks at time $t_*.$ Using (\ref{transf}) these probability assignments can be transformed into 
probability densities $p_\pm(\tau).$  Note taking the $\tau$-derivative of (\ref{transf}) gives
$$ \dot{b}_\pm(\tau) = \pm \frac{\frac{1}{2}(u_-(\tau)-u_+(\tau))}{1+u'(b_\pm(\tau),t_0)(\tau-t_0)}
= \pm \frac{1}{2}(u_-(\tau)-u_+(\tau)) (1-u'_\pm(\tau)(\tau-t_0)) , $$
with $u'_\pm(\tau)=u'(x_*(\tau)\pm 0,\tau).$ Since also $b_+^f-b_-^f=(u_-^f-u_+^f)(t_f-t_0)$
we obtain the unique assignment
\be p_\pm(\tau)=\frac{(u_-(\tau)-u_+(\tau)) (1-u'_\pm(\tau)(\tau-t_0))}{2(u_-^f-u_+^f)(t_f-t_0)}, 
\,\,\,\, \tau\in [t_*,t_f]. \lb{ppm} \ee
which completely specifies the process. 

This random process is Markov in an extended state space $X(\tau)\subset \RR \times \{-1,0,1\},$ 
which is defined, precisely, by 
$$ X(\tau)=\{(x,-1):\ x\leq x_*(\tau)\} \bigcup \{(x,0):\ x=x_*(\tau)\} \bigcup \{(x,+1):\ x\geq x_*(\tau)\} $$
for $\tau\geq t_*,$ with the three subsets denoted $X_{-1}(\tau),X_0(\tau),X_{+1}(\tau),$ respectively. 
Likewise, 
$$ X(\tau)=\{(x,-1):\ x<a_*(\tau)\} \bigcup \{(x,0):\ x=a_*(\tau)\} \bigcup \{(x,+1):\ x> a_*(\tau)\} $$
for $\tau<t_*$ with $x_*(\tau)$ replaced by the particle position $a_*(\tau)$ at time $\tau$ which evolves 
into the shock at time $t_*.$ The time-dependent infinitesimal generator $L(\tau)$ of the process is
$$ L(\tau)f(x,\pm 1) = -u(x,\tau)f'(x,\pm 1), \ x\in X_{\pm 1}(\tau) $$
$$ L(\tau)f(x,0) = -u_*(\tau)f'(x,0) + \sum_{\alpha=\pm 1} \lambda_\alpha(\tau) [f(x,\alpha)-f(x,0)], \ x\in X_{0}(\tau) $$
for $\tau\geq t_*,$ and
$$ L(\tau)f(x,\alpha) = -u(x,\tau)f'(x,\alpha), \ (x,\alpha)\in X(\tau) $$
for $\tau<t_*.$ Here we have used the definitions of jump rates to the right/left off the shock as
$$ \lambda_\pm(\tau) = p_\pm(\tau)/P(\tau), \,\,\,\, P(\tau)=\int_{t_*}^\tau \rmd t \ [p_+(t)+p_-(t)]. $$
The factor $P(\tau)$ is the probability that the particle remains on the shock at time $\tau$ backward 
in time and appears because of the definition of the generator through a conditional expectation. 
Consistent with the fact that the particle must leave the shock by time $t_*$ a.s., 
$\lim_{\tau\rightarrow t_*+} \lambda_\pm(\tau)=+\infty.$ 
Notice that this is a jump-drift process in the extended state space but, projected down to $\RR,$
the realizations $x(t)$ are continuous functions of time $t$. On the other hand, the velocity process 
$D_t^+x(t)=\bar{u}(x(t),t)$---which we denote for simplicity as $\dot{x}(t)$---is only right-continuous 
with jump discontinuities at discrete times.  

We next establish an important property for this process: 
\begin{proposition} 
$$ {\mathbb E}(\dot{x}(t)) = u_*^f \,\,\,\, \mbox{for all $t\in [0,t_f].$} $$
\end{proposition}\lb{fun} 

\noindent
{\it Proof:} This is obvious for $t=t_0$ since, by construction,
\be {\mathbb E}(\dot{x}(t_0)) = \frac{1}{b_+^f-b_-^f}\int_{b_-^f}^{b_+^f} \rmd b \ u(b,t_0) = u_*^f. \lb{fund-tstar} \ee

Next consider $0\leq t<t_0.$ Since all of the realizations of the process are smooth solutions of 
Burgers for $t<t_0,$ we have a smooth invertible relation between the Lagrangian positions at times $t_0$ and $t$:
\be  b = c + u(c,t) (t_0-t), \,\,\,\, u(b,t_0) = u(c,t). \lb{backmap} \ee
The uniform distribution on $b\in [b_-^f,b_+^f]$ is transformed into the distribution with density 
\be p(c,t) = \frac{1 + u'(c,t)(t_0-t) }{b_+^f-b_-^f} \lb{backpdf} \ee
on the shock interval $[c_-^f,c_+^f]$ at time $t.$ Note this density is nonnegative and 
$\int_{c_-^f}^{c_+^f} \rmd c \ p(c,t) =1$ using $b_+^f-b_-^f=c_+^f-c_-^f+(u_+^f-u_-^f)(t_0-t).$
We can immediately infer that for $0\leq t<t_0.$
\begin{eqnarray} 
{\mathbb E}(\dot{x}(t)) &=& \int_{c_-^f}^{c_+^f} \rmd c \ u(c,t) \ p(c,t) \cr
                                    &=& \frac{1}{b_+^f-b_-^f}\left[\int_{c_-^f}^{c_+^f} \rmd c \ u(c,t) 
                                             + \frac{1}{2}(|u_+^f|^2-|u_-^f|^2)(t_0-t)\right] \cr
                                    &=&  \frac{1}{c_+^f-c_-^f}\int_{c_-^f}^{c_+^f} \rmd c \ u(c,t) = u_*^f        
\end{eqnarray} 
using in the second line $\frac{1}{2}(u_+^f+u_-^f)=\frac{1}{c_+^f-c_-^f} \int_{c_-^f}^{c_+^f} \rmd c \ u(c,t)$ 
to obtain the third line. Although we have verified these properties by explicit calculations with (\ref{backpdf}),
they indeed follow directly by its definition. In particular, the average of $u(c,t)$ for $p(c,t)$ must coincide 
by (\ref{backmap}) with the uniform average of $u(b,t_0)$ over $[b_-^f,b_+^f].$ 

\begin{figure}[!ht]
\begin{center}
\includegraphics[height=3in,width=3.5in]{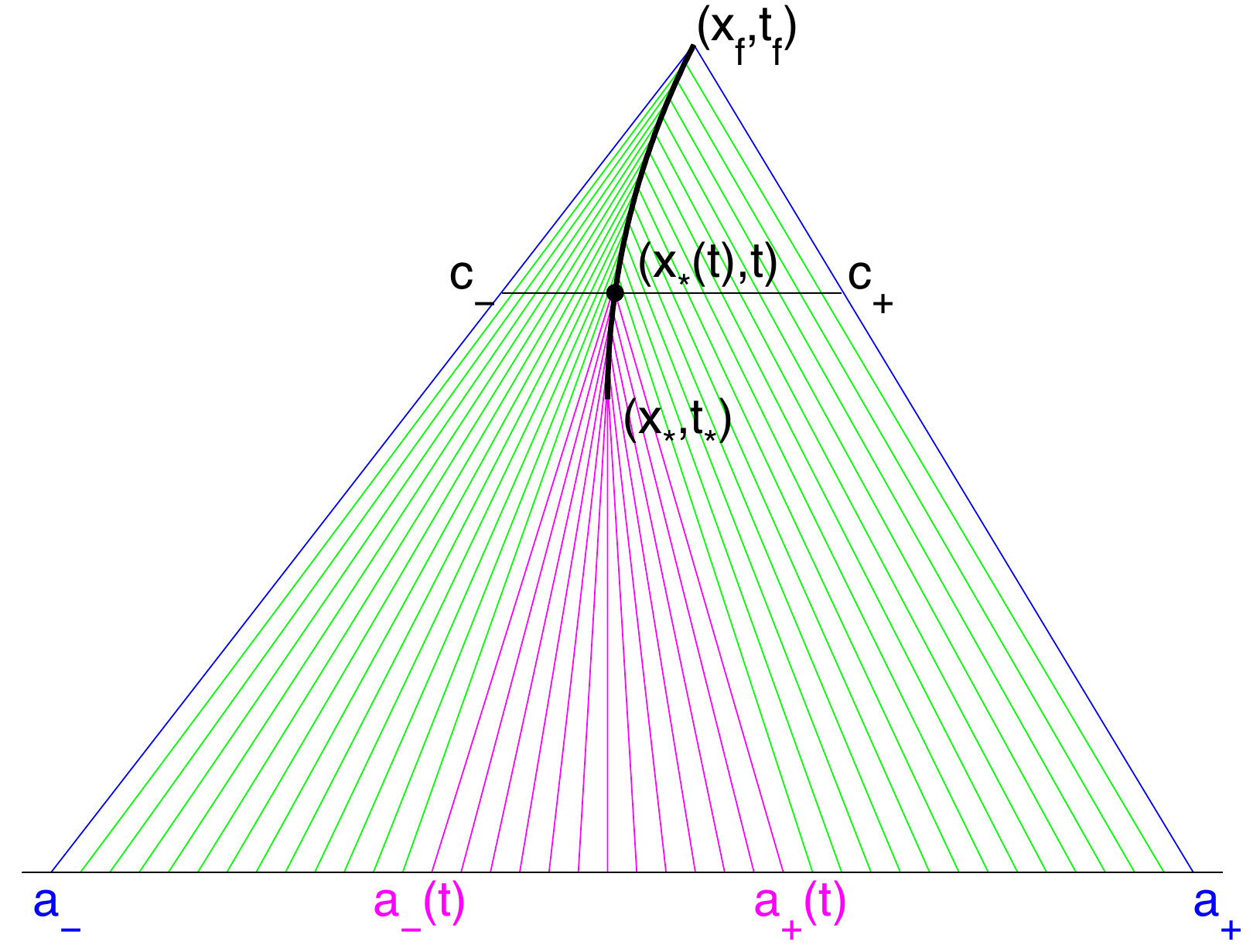}
\caption{{\small {\it Particle Positions at Time $t>t_*$.} The horizontal black line
shows the set of possible positions $[c_-,c_+]$ at time $t$ of particles in the shock interval. 
There is an atom of positive probability situated on the shock itself, indicated by the black dot.
The probability of this atom is the fraction of the interval $[a_-^f,a_+^f]$ inside the 
subinterval $[a_-(t),a_+(t)]$ which maps into the shock at time $t.$}}
\end{center}
\end{figure}

Finally, we consider the case $t_0<t<t_f.$ When $t_0<t<t_*,$ the analysis is similar to that above, except 
that particle labels $c$ at time $t$ are related to those at time $t_0$ by
\be c = b + u(b,t_0) (t-t_0), \,\,\,\, u(c,t) = u(b,t_0), \lb{foremap} \ee
where $c\in [c_-^f,c_+^f]$. A uniform distribution on $[b_-^f,b_+^f]$ implies a density 
\be p(c,t) = \frac{1-u'(c,t)(t-t_0)}{b_+^f-b_-^f}  \lb{forepdf} \ee
which is nonnegative, as seen from $1-u'(c,t)(t-t_0)=[1+u'(b,t)(t-t_0)]^{-1},$ and integrates to 1 and has 
mean value $u_f^*.$ However, when $t_*<t<t_f,$ the situation is different, as illustrated in Fig.~3 for the 
special choice $(b,t_0)=(a,0).$ Now there are two components to the probability distribution 
of $x(t),$ a continuous component and an atom of positive probability associated to the shock.  
The continuous component has total probability 
\be \int_t^{t_f} \rmd \tau \ [p_+(\tau)+p_-(\tau)] = \frac{1}{b_+^f-b_-^f}\int_{[b_-^f,b_-(t)]\cup [b_+(t),b_+^f]} \rmd b 
    = 1- \left(\frac{b_+(t)-b_-(t)}{b_+^f-b_-^f}\right) \lb{norm2} \ee
and the atom has the complementary probability $\frac{b_+(t)-b_-(t)}{b_+^f-b_-^f}$.  In Fig.~3 this latter probability 
corresponds to the relative length of the subinterval $[a_-(t),a_+(t)]$ inside $[a_-^f,a_+^f],$ which is mapped by 
the forward coalescing flow into the atom (magenta lines). The continuous component is again given by 
(\ref{foremap}),(\ref{forepdf}), but now $b\in [b_-^f,b_-(t)]\cup [b_+(t),b_+^f]$ and the continuous density
integrates to the value in the right-hand side of (\ref{norm2}). The last statement follows because the transformation 
$c \leftrightarrow b$ is smooth and invertible between $[c_-^f,c_+^f]$ and $[b_-^f,b_-(t)]\cup [b_+(t),b_+^f],$ or 
can be checked by explicit calculation. It likewise follows that the contribution of the continuous component 
to ${\mathbb E}(\dot{x}(t))$ is now:
$$ \int_{c_-^f}^{c_+^f} \rmd c \ u(c,t) \ p(c,t) = \frac{1}{b_+^f-b_-^f}\int_{[b_-^f,b_-(t)]\cup [b_+(t),b_+^f]} \rmd b \ u(b,t). $$
The contribution of the atom to ${\mathbb E}(\dot{x}(t))$, on the other hand, is
$$  \left(\frac{b_+(t)-b_-(t)}{b_+^f-b_-^f}\right)\cdot u_*(t) = 
       \frac{1}{b_+^f-b_-^f}\int_{[b_-(t),b_+(t)]} \rmd b \ u(b,t), $$
where we have used the fundamental property (\ref{avrgcond}) that 
$$ u_*(t) = \frac{1}{b_+(t)-b_-(t)}\int_{b_-(t)}^{b_+(t)} \rmd b \ u(b,t). $$       
Finally, adding the two contributions gives
$$ {\mathbb E}(\dot{x}(t)) =  \frac{1}{b_+^f-b_-^f} \int_{b_-^f}^{b_+^f} \rmd b \ u(b,t) = u_*^f. $$
\hfill $\Box$

An elaboration of these arguments shows furthermore that
\begin{proposition}
$$ {\mathbb E}\big(\dot{x}(t)\big|\dot{x}(s)\big) = \dot{x}(s) \,\,\,\, \mbox{for all $t\leq s,$ $t,s\in [0,t_f].$} $$
\end{proposition}\lb{fun2} 

\noindent {\it Proof:} We discuss separately the cases where $\dot{x}(s)=u_*(s)$ and $\dot{x}(s)\neq u_*(s).$

When $u\neq u_*(s),$ then the set $\{\dot{x}(s)=u\}$ has zero probability and the conditional 
distribution is supported on smooth characteristics for Burgers solutions which pass through 
the finite number of points $y_i$ which satisfy $u(y_i,s)=u$. Since the velocity is preserved 
along characteristics, the conditional distribution can likewise be defined as that supported 
on smooth characteristics which pass through the finite number of points $x_i$ which satisfy 
$u(x_i,t)=u,$ with the two sets of points related by $y_i=x_i+u\cdot (s-t).$ The conditional distribution 
$ p(x,t|\dot{x}(s)=u)$ thus has the standard definition 
$$ p(x,t|\dot{x}(s)=u) = \frac{p(x,t)\delta (u(x,t)=u)}{\int \rmd x \ p(x,t)\delta (u(x,t)=u)}. $$
But 
$$ \delta (u(x,t)=u) = \sum_i \frac{\delta (x-x_i)}{|u'(x_i,t)|}, $$
so that 
$$  p(x,t|\dot{x}(s)=u) = \sum_i w_i \delta(x-x_i) $$
with $w_i=\frac{p(x_i,t)}{N|u'(x_i,t)|}$ and $N=\sum_i \frac{p(x_i,t)}{|u'(x_i,t)|}.$ However, since $u(x_i,t)=u$
for all $i,$ it is immediate that 
$$ {\mathbb E}(\dot{x}(t)\big|\dot{x}(s)=u) = \int \rmd x \ u(x,t) p(x,t|\dot{x}(s)=u) = \sum_i w_i u(x_i,t) = u. $$

On the other hand, when $u=u_*(s),$ then the set $\{\dot{x}(s)=u\}$ has positive probability 
and the conditional distribution is supported on the generalized characteristics which are at $x_*(s)$
at time $s.$ Note that there may be smooth characteristics not in the shock at time $s$ which also 
happen to have the same velocity $u,$ but these have zero probability and may be neglected. 
By the construction of the original random process, the conditional distribution is identical to that obtained 
by assuming a uniform distribution on the interval of Lagrangian positions $[b_-(s),b_+(s)]$ at time $t_0$ 
which map into $x_*(s)$ at time $s.$ See Fig.~3 for $(b,t_0)=(a,0),$ changing $t$ there to $s$. Note that 
if $t_*<t<s,$ then this distribution has an atom located at the shock and if $t<t_*$ then it has only a continuous part.
In either case 
$$ {\mathbb E}(\dot{x}(t)\big|\dot{x}(s)=u) = \dot{x}_*(s) = u $$
by the arguments used in proving Proposition \ref{fun2}.1. 
\hfill $\Box$  

\vspace{20pt}
\noindent {\bf Remark \# 1:} By the Markov properties of the process,
$$ {\mathbb E}\big(\dot{x}(t)\big|\{\dot{x}(\tau),\ \tau\geq s\}\big) = {\mathbb E}\big(\dot{x}(t)\big|\dot{x}(s)\big). $$
Thus, the velocity process is a {\it backward martingale}. It is also a backward martingale
with respect to the position process $x(t),$ that is, 
$$ {\mathbb E}\big(\dot{x}(t)\big|\{x(\tau),\ \tau\geq s\}\big) = {\mathbb E}\big(\dot{x}(t)\big|x(s)\big)=\dot{x}(s). $$

\vspace{20pt}
\noindent {\bf Remark \# 2:} Proposition \ref{fun2}.2 in fact implies Proposition \ref{fun}.1 by choosing $s=t_f$ and 
noting that $\dot{x}(t_f)=u_*^f$ with probability one, so that 
$$ {\mathbb E}(\dot{x}(t)) = {\mathbb E}\big(\dot{x}(t)\big|\dot{x}(t_f)=u_*^f\big) = u_*^f. $$

\vspace{20pt}
\noindent {\bf Remark \# 3:} The generalization of the above geometric construction to Burgers solutions 
with countably many shocks is straightforward, but the details are a bit tedious. We outline the 
multi-shock construction in Appendix A. 

\vspace{20pt}
\noindent {\bf Remark \# 4:} 
As the present paper was being prepared for publication, we became aware of an interesting recent 
work of Moutsinga \cite{Moutsinga12}, whose results are closely related to those of the present 
section. Moutsinga's goal was to derive the entropy solution of inviscid Burgers equation from a suitable 
sticky particle model corresponding to a forward coalescing flow with initial uniform mass density 
$\rho_0(db)=db$ (Lebesgue measure).   
His Theorem 2.1 gave the time-evolved mass measure as 
$$ \rho_t(dc)=dc - (t-t_0) du(c,t) $$
where $du(c,t)$ is the measure defined by the Lebesgue-Stieltjes integral with respect to the Burgers 
solution $u(\cdot,t).$ This result coincides with our formulas (\ref{forepdf}),(\ref{norm2}), except 
that our ``mass densities'' are normalized to be probability measures. Furthermore, Moutsinga's 
Theorem 2.2 gave the martingale property in our Proposition 3.2
as 
$$ u(\xi_{s,t},t) = {\mathbb E}_{\rho_s} \left[\left. u(\cdot,s)\right| \xi_{s,t}\right], \,\,\,\, \rho_s-{\rm a.e.}  
\,\,\,\,\,\,\, t_0\leq s\leq t, $$
when transcribed into our notations. Moutsinga's result is itself a generalization of earlier such theorems  
for forward coalescing flows in sticky particle models of pressure-less gas dynamics \cite{Dermoune99}. 
The main innovation in our work here is to point out the existence of a stochastic inverse of the forward 
flow which is Markov backward in time and under which the fluid particle velocities are backward martingales.

\section{Zero-Viscosity Limit}\lb{sec:zerovisc}  

We now construct a fundamentally different stochastic inverse by considering the zero-viscosity 
limit of the stochastic representation of Constantin \& Iyer \cite{ConstantinIyer08} for the 
viscous Burgers solutions. 

\subsection{The Constantin-Iyer Representation}\lb{sec:CI} 

To make our discussion self-contained, we begin by presenting a new derivation of the Constantin-Iyer
representation for viscous Burgers solutions. We then establish the close relation of this stochastic 
representation to the classical Hopf-Cole representation. These results hold for any space dimension 
$d\geq 1,$ so that we discuss in this section multi-dimensional Burgers. 

Consider a solution $\bu$ to the viscous Burgers equation on the space-time domain $D=\RR^d \times [t_0,t_f]$ with initial condition $\bu_0(\bx)$ and define 
the {\it backward stochastic flow} $\tbxi_{t,s}$ for $t_f\geq t\geq s\geq t_0$ by the solution of the stochastic differential equation
\be \rmd \tbxi_{t,s}(\bx) = \bu(\tbxi_{t,s}(\bx),s) \rmd s + \sqrt{2\nu}\, \hd \tilde{\bW}(s) \label{stochflow1} \ee
with final conditions 
\be \tbxi_{t,t}(\bx)=\bx, \,\,\,\, \bx\in \mathbb{R}^d, \,\, t\in [t_0,t_f]. \label{stochflow2} \ee  
Here $\tilde{\bW}(t)$ denotes an $\RR^d$-valued Wiener process and ``$\hd$'' in (\ref{stochflow1}) implies a backward Ito SDE. 
These flows enjoy the semigroup property $\tbxi_{s,r}\circ\tbxi_{t,s}=\tbxi_{t,r}$ a.s. for $t\geq s\geq r.$ 
For the basic theory of backward It$\bar{\rm o}$ integration and stochastic flows that we use below, 
see \cite{Friedman06, Kunita97}. 

The fundamental property of the backward stochastic flows defined above for the viscous Burgers velocity 
field is given by the following:
\begin{proposition}\label{martingale}
The stochastic Lagrangian velocity $\tbv(s|\bx,t)=\bu(\tbxi_{t,s}(\bx),s)$ is a backward martingale of  
the stochastic flow defined by (\ref{stochflow1}), (\ref{stochflow2}). That is, 
$$ \mathbb{E}\big(\tbv(s|\bx,t)\big|\mathcal{F}_{t,r}\big)=\tbv(r|\bx,t), \,\,\,\, t\geq r\geq s,$$
where $\mathcal{F}_{t,r}$ is the filtration of sigma algebras $\sigma\{\tilde{\bW}(u): t\geq u\geq r\}.$ 
\end{proposition}
\begin{proof}
By the backward It$\bar{\rm o}$ formula for flows \cite{Kunita97}, we have, for any $\bx\in \mathbb{R}^d$ and for each $t,t'$ satisfying $t_0\leq t'\leq t<t_f$ that 
$$
\rmd \bu( \tbxi_{t,t'}(\bx),t') = \bu_t \rmd t'+  (\hd \tbxi_{t,t'}\cdot \grad_x)\bu  -\frac{1}{2} \bu_{x_ix_j} \rmd \langle \txi_{t,t'}^i, \txi_{t,t'}^j\rangle,
$$
where the quadratic variation can be calculated from (\ref{stochflow1}) to be $ \rmd \langle \txi_{t,t'}^i, \txi_{t,t'}^j\rangle = 2\nu \delta^{ij}\rmd t'. $
Thus, 
\begin{align*}
\rmd \bu( \tbxi_{t,t'}(\bx),t') &= \left.\left(\left(\bu_t+(\bu\cdot\grad)\bu -\nu \triangle\bu\right) \rmd t'+ \sqrt{2\nu} \ \hd  \tilde{\bW}(t')\cdot \grad_x\bu \right) \right|_{ (\tbxi_{t,t'}(\bx),t')}\\
&= \sqrt{2\nu} \ \hd  \tilde{\bW}(t') \cdot \grad_x\bu(\tbxi_{t,t'}(\bx),t') 
\end{align*}
In passing to the second line, we used the fact that $\bu$ solves Burgers equation on $D$. Integrating over $t'\in [s,r]$ gives 
$$ \bu( \tbxi_{t,s}(\bx),s) = \bu( \tbxi_{t,r}(\bx),r) + \sqrt{2\nu}\int_s^r \ \hd  \tilde{\bW}(t') \cdot \grad_x\bu(\tbxi_{t,t'}(\bx),t') $$
Since $\mathbb{E}\Big( \int_s^r \hd  \tilde{\bW}(t') \cdot \grad_x\bu(\tbxi_{t,t'}(\bx),t')\Big|\mathcal{F}_{t,r}\Big)=0$ for a backward It$\bar{\rm o}$ integral, the result 
follows. \hfill $\Box$
\end{proof}

\noindent Note that unconditional expectation gives
\be \mathbb{E}\big(\bu(\tbxi_{t,s}(\bx),s)\big)=
  \mathbb{E}\big(\bu(\tbxi_{t,s}(\bx),s)\big|\mathcal{F}_{t,t}\big)=\bu(\bx,t). \lb{mart2} \ee
This leads to the Constantin-Iyer representation: 
\begin{proposition}\label{burgersCI}
 A smooth function $\bu$ on the space-time domain $D=\RR^d \times [t_0,t_f]$ is a solution to the viscous Burgers equation 
with initial data $\bu(\cdot,t_0)=\bu_0$ if and only if for each $(\bx,t)\in D$ it satisfies
\begin{align}
\bu(\bx,t)&= \mathbb{E}\left[\bu_0 (\tbxi_{t,t_0}(\bx))\right]  \label{burgersvel},
\end{align}
where the map $\tbxi_{t,s}$ is the stochastic flow defined by (\ref{stochflow1}), (\ref{stochflow2}) for the velocity field $\bu$.
\end{proposition}

\begin{proof}
First suppose that $\bu$ solves the viscous Burgers equation with initial condition $\bu_0$. Using (\ref{mart2}) with $s=0$ yields formula \eqref{burgersvel}.  
Now for the other direction, assume that \eqref{burgersvel} holds, together with (\ref{stochflow1}),(\ref{stochflow2}). Using the semigroup property 
of the stochastic flow maps and the $\mathcal{F}_{t,s}$-measurability of $\tbxi_{t,s}$ gives for any $t,t'\in [t_0,t_f],\,\, t\geq t' $
\begin{align*}
\mathbb{E}\left[\bu_0(\tbxi_{t,0}(\bx))\right] 
&= \mathbb{E}\left[\bu_0\left(\tbxi_{t',0}\circ \tbxi_{t,{t'}}(\bx)\right)\right] \cr
& = \mathbb{E}\left[\mathbb{E}\left[\bu_0\left(\tbxi_{t',0}\left(\tbxi_{t,{t'}}(\bx)\right)\right)\Big|\mathcal{F}_{t,t'}\right] \right] \cr
&= \mathbb{E}\left[\mathbb{E}\left.\left[\bu_0\left(\tbxi_{t',0}(\by)\right)\right]\right|_{\tbxi_{t,{t'}}(\bx)=\by} \right] \ = \ \mathbb{E}\left[\bu(\tbxi_{t,{t'}}(\bx),t')\right]. 
\end{align*}
We therefore see that the following equivalent representation is implied for $t\geq t'$:
\begin{align*}
\bu(\bx,t)&= \mathbb{E}\left[\bu(\tbxi_{t,{t'}}(\bx),t')\right].  
\end{align*}
An application of the backward It$\bar{\rm o}$'s formula to $\bu \circ \tbxi_{t,t'}$ gives:
\begin{align*}
\bu(\bx,t)= \ & \bu(\tbxi_{t,{t'}}(\bx),t')+ \int_{t'}^t\left. \left(\partial_s \bu+(\bu\cdot \grad)\bu -\nu \triangle \bu\right) \right|_{(\tbxi_{t,s}(\bx),s)} \ \rmd s
+ \sqrt{2\nu}\int_{t'}^t  \ \hd \tilde{\bW}_s \cdot \grad_x \bu( \tbxi_{t,s}(\bx),s)
\end{align*}
Using the above results for $t>t'$ and computing 
\begin{align*}
0&=\frac{\bu(\bx,t)- \mathbb{E}\left[\bu(\tbxi_{t,{t'}}(\bx),t')\right] }{t-t'}=
\mathbb{E}\left[  \frac{\int_{t'}^t\left.\left(\partial_s \bu_s+\bu_s\cdot \grad\bu_s -\nu \triangle \bu_s \right) \right|_{ (\tbxi_{t,s}(\bx),s)} \rmd s }{t-t'}\right]
\end{align*}
which, at the coincidence limit ${t'}\nearrow t$, proves that any $\bu$ satisfying (\ref{burgersvel}) must solve the viscous Burgers equation.
\hfill $\Box$
\end{proof} 

The C-I representation (\ref{burgersvel}) for Burgers is exactly analogous to that employed for studies of passive 
scalar advection the Kraichnan model \cite{Bernardetal98,GawedzkiVergassola00} and can be written also as 
$$ \bu(\bx,t) = \int \rmd^d {\rm a} \ \bu(\ba,s) \ p_\bu(\ba,s|\bx,t), \,\,\,\, s\leq t $$
where
$$ p_\bu(\ba,s|\bx,t) = {\mathbb E}\left[\delta^d(\ba-\tbxi_{t,s}(\bx))\right] $$
is the transition probability for the backward diffusion. Unlike the linear relation for the Kraichnan model, however, 
the C-I representation is a nonlinear fixed point condition for the Burgers solution, because the drift of the diffusion process 
is the Burgers velocity itself. There should be close connections with the stochastic control formulation introduced 
by P.-L. Lions for general Hamilton-Jacobi equations \cite{Lions83,FlemingSoner05}. Note, however, 
that the C-I representation requires no assumption that the velocity field is potential. When this latter condition holds 
it is possible to establish a direct relation with the Hopf-Cole solution \cite{Hopf50,Cole51}, by means of the following: 


\begin{proposition} 
The backward transition probabilities $p_\bu(\ba,s|\bx,t)$ of the stochastic Lagrangian trajectories in the C-I representation 
of viscous Burgers equation have the form 
\begin{align}
p_\bu(\ba,s| \bx,t)&= \frac{\exp\left(-\frac{1}{2\nu}S(\ba,s|\bx,t) \right) }{\int_{\RR^d} 
\exp\left(-\frac{1}{2\nu}S(\ba',s|\bx,t)  \right)\rmd^d {\rm a}'} \label{Sprob} 
\end{align}
with 
\be  S(\ba,s|\bx,t) = \frac{|\bx-\ba|^2}{2(t-s)} + \phi(\ba,s)-\phi(\bx,t), \lb{Sdef} \ee
and $\phi$ is any solution of the KPZ/Hamilton-Jacobi equations 
\begin{align}\label{phieq} 
\partial_t \phi +\frac{1}{2}|\grad\phi |^2 &= \nu \triangle \phi +\gamma(t),\\
\phi(\bx,t_0)&=\phi_0(\bx)+c_0, \nonumber
\end{align}
where $ \bu_0=\grad\phi_0$ but the function $\gamma(t)$ and constant $c_0$ may be freely chosen. 
\end{proposition}\lb{Drivasform} 

\begin{proof}
Calculate the transition probability by the Girsanov transformation
\begin{align*}
p_\bu(\ba,s| \bx,t) & =\mathbb{E}^W\left[\delta^d(\tbxi_{t,s}(\bx)-\ba) \right]\\
& =\mathbb{E}_\bx^{\xi,\nu}\left[\delta(\tbxi_{t,s}(\bx)-\ba)\left(\frac{\rmd \bP^{W} }{\rmd \bP^{\xi,\nu}_\bx } \right)\right],
\end{align*}
where the first expectation $\mathbb{E}^W$ is over the Wiener measure $\mathcal{P}^W$ associated to the Brownian motion $\tilde{\bW},$ the second expectation 
is over the (scaled) Wiener measure $\mathcal{P}^{\xi,\nu}_\bx$ associated to the Brownian motion $\tbxi_{t,s}(\bx)\sim \bx + \sqrt{2\nu}\tilde{\bW}(s),$ and the Radon-Nikodym
derivative (change of measure) is given by the backward Girsanov formula:
\begin{align*}
\frac{\rmd \bP^{W} }{\rmd \bP^{\xi,\nu}_\bx } =\exp\left[\frac{1}{2\nu}\int_s^t \left(\bu(\tbxi_{t,\tau}(\bx),\tau) \cdot \hd \tbxi_\tau - \frac{|\bu(\tbxi_{t,\tau}(\bx),\tau)|^2}{2}\rmd \tau \right)  \right]
\end{align*}
Now, suppose we are considering potential flow so that $\bu(\bx,t)= \grad\phi(\bx,t)$.  Demanding that $\bu$ satisfy the viscous Burgers equation, $\phi$ must satisfy 
(\ref{phieq}). By the backward It$\bar{\rm o}$ formula we have that 
\begin{align*}
\rmd \phi(\tbxi_{t,\tau}(\bx),\tau) &= \left(\partial_\tau \phi - \nu \triangle \phi\right)\big|_{(\tbxi_{t,\tau}(\bx),\tau)}\rmd \tau + \hd \tbxi_{t,\tau} \cdot \grad\phi(\tbxi_{t,\tau}(\bx),\tau)\\
&= \gamma(\tau)\rmd \tau -\frac{1}{2}|\grad \phi(\tbxi_{t,\tau}(\bx),\tau)|^2\rmd \tau+ \hd \tbxi_{t,\tau} \cdot \grad\phi(\tbxi_{t,\tau}(\bx),\tau)
\end{align*}
The Girsanov formula thus becomes:
\begin{align*}
\frac{\rmd \bP^{W} }{\rmd \bP^{\xi,\nu}_\bx }= \frac{1}{\mathcal{N}}\exp\left(\frac{1}{2\nu}\left( \phi(\bx,t)-\phi(\tbxi_{t,s}(\bx),s) \right) \right)
\end{align*}
where $\mathcal{N}$ must be chosen to satisfy the normalization condition $\mathbb{E}_x^{\xi,\nu}\left[\frac{\rmd \bP^{W} }{\rmd \bP^{\xi,\nu}_\bx }\right]=1.$
Using the equality in distribution $\tbxi_{t,s}(\bx)\sim \bx + \sqrt{2\nu}\tilde{\bW}(s),$ one obtains
\begin{align*}
p_\bu(\ba,s| \bx,t)& =\frac{1}{\mathcal{N}}\mathbb{E}_\bx^{\xi,\nu}\left[\delta(\tbxi_{t,s}(\bx)-\ba) \exp\left(\frac{1}{2\nu}\left( \phi(\bx,t)-\phi(\tbxi_{t,s}(\bx),s) \right) \right) \right]\\
& =\frac{1}{\mathcal{N}}\mathbb{E}_\bx^{\xi,\nu}\big[\delta(\tbxi_{t,s}(\bx)-\ba) \big]\exp\left(\frac{1}{2\nu}\left( \phi(\bx,t)-\phi(\ba,s) \right) \right) \\
&= \frac{1}{\mathcal{N}}\frac{1}{(4\pi \nu t)^{d/2}}\exp\left(-\frac{|\bx-\ba|^2}{4\nu(t-s)}+ \frac{1}{2\nu}\left( \phi(\bx,t)-\phi(\ba,s) \right) \right) 
\end{align*}
with
$$
\mathcal{N}=\frac{1}{(4\pi \nu t)^{d/2}}\int_{\RR^d} \exp\left(-\frac{|\bx-\ba|^2}{4\nu(t-s)}+ \frac{1}{2\nu}\left( \phi(\bx,t)-\phi(\ba,s) \right) \right) \rmd^d {\rm a}.
$$
\hfill $\Box$
\end{proof}

\vspace{20pt}
\noindent {\bf Remark \# 1:} 
For related calculations using a forward Girsanov transformation, see \cite{Garbaczewskietal97}.
The backward Girsanov formula is equivalent to the Lagrangian path-integral 
$$ p_\bu(\ba,s| \bx,t) = \int_{\bx(t)=\bx} \mathcal{D}\bx \ \delta^d(\bx(s)-\ba)\exp\left(-\frac{1}{4\nu}\int_s^t |\dot{\bx}(\tau)-\bu(\bx(\tau),\tau)|^2 \ \rmd \tau\right), $$
which appears in the physics literature \cite{ShraimanSiggia94, Falkovichetal01}. For a careful discussion of this equivalence, see the Appendix of 
\cite{Eyink11}.

\vspace{20pt}
\noindent {\bf Remark \# 2:} 
It is now straightforward to show that the C-I representation is equivalent to the Hopf-Cole formula \cite{Hopf50,Cole51}:
$$ \bu(\bx,t) =
-2\nu \grad_x\ln\left[\frac{1}{(4\pi\nu t)^{d/2}}\int_{\RR^d} \exp\left(-\frac{|\bx-\ba|^2}{4\nu t}-\frac{\phi_0(\ba)}{2\nu}\right) \rmd^d{\rm a} \right]. $$
Using the chain rule and moving the gradient inside the integral gives 
\begin{eqnarray}
\bu(\bx,t) &=& 
\frac{\frac{1}{(4\pi\nu t)^{d/2}}\int_{\RR^d} 2\nu\grad_a \exp\left(-\frac{|\bx-\ba|^2}{4\nu t}\right)\cdot \exp\left(-\frac{\phi_0(\ba)}{2\nu}\right) \rmd^d{\rm a}} 
{\frac{1}{(4\pi\nu t)^{d/2}}\int_{\RR^d} \exp\left(-\frac{|\bx-\ba|^2}{4\nu t}-\frac{\phi_0(\ba)}{2\nu}\right) \rmd^d{\rm a}}\cr
&=& \frac{\int_{\RR^d}  \exp\left(-\frac{|\bx-\ba|^2}{4\nu t}-\frac{\phi_0(\ba)}{2\nu}\right) \grad_a \phi_0(\ba) \rmd^d{\rm a}} 
{\int_{\RR^d} \exp\left(-\frac{|\bx-\ba|^2}{4\nu t}-\frac{\phi_0(\ba)}{2\nu}\right) \rmd^d{\rm a}}\cr
&=& \int_{\RR^d} \bu_0(\ba)\ p_\bu(\ba,0|\bx,t)  \rmd^d{\rm a}
\end{eqnarray}
using integration by parts, $\bu_0(\ba)=\grad\phi_0(\ba),$ and the expression:
\begin{align}
p_\bu(\ba,s| \bx,t)&= \frac{\exp\left(-\frac{|\bx-\ba|^2}{4\nu(t-s)}- \frac{1}{2\nu}\phi(\ba,s)  \right) }{\int_{\RR^d} \exp\left(-\frac{|\bx-\ba'|^2}{4\nu(t-s)}- \frac{1}{2\nu}\phi(\ba',s)  \right)\rmd^d {\rm a}'}
\end{align}
with some common factors canceled.

\subsection{Spontaneous Stochasticity}\lb{sec:spontstoch} 

We now employ the results of the previous section to show that the backward diffusion process associated 
to the C-I representation remains random (non-deterministic) as $\nu\rightarrow 0,$ which is exactly 
the property of spontaneous stochasticty (backward in time). We here explicitly denote the viscosity 
dependence of the Burgers solution by subscript, as $\bu_\nu=\grad\phi_\nu,$ and the zero-viscosity limit 
is denoted by $\bu_*=\grad\phi_*.$ We also define the measure on $\RR^d$ 
$$P_\nu^{s;\bx,t}(\rmd\ba) = \rmd^d {\rm a} \ p_{\bu_\nu}(\ba,s|\bx,t) $$
associated to the backward diffusion with densities (\ref{Sprob}). Our limit result is then stated as:  

\begin{proposition}
For any sequence $\bx_\nu= \bx + O(\nu),$ the probability measures $P_\nu^{s;\bx_\nu,t}$ on $\RR^d$ 
for each $\nu>0$ converge weakly along subsequences in the limit $\nu\rightarrow 0$ 
to probability measures $P_*^{s;\bx,t},$ which may depend on the subsequence but which are always supported 
on atoms in the finite set 
$$ {\mathcal A}_{s;\bx,t} = {\rm argmin}_\ba \left[ \frac{|\bx-\ba|^2}{2(t-s)}+\phi_*(\ba,s)\right].$$
If $(\bx,t)$ is a regular point 
of the limiting inviscid Burgers solution $\bu_*,$ then $P_\nu^{s;\bx_\nu,t} \stackrel{w}{\longrightarrow} P_*^{s;\bx,t}=\delta_{\ba(s;\bx,t)},$ where 
$\ba(s;\bx,t)=\bx-\bu_*(\bx,t)(t-s)$ is the inverse Lagrangian image at time $s<t$ of $\bx$ at time $t.$
Suppose instead that $(\bx,t)$ is a generic point on the shock set of $\bu_*$ and the sequence $\bx_\nu$ 
satisfies 
\be \lim_{\nu\rightarrow 0} \bu_\nu(\bx_\nu,t) = p\ \bu_*(\bx^-,t) + (1-p) \bu_*(\bx^+,t), \,\,\,\,
p\in [0,1] \lb{p-lim} \ee
where the velocities $\bu(\bx^\pm,t)$ are the limits from the two sides of the shock. Then
\be P_\nu^{s;\bx_\nu,t} \stackrel{w}{\longrightarrow} P_*^{s;\bx,t}=p\ \delta_{\ba_+(s;\bx,t)} + (1-p) \delta_{\ba_-(s;\bx,t)} \lb{2delta} \ee
where $\ba_\pm(s;\bx,t)=\bx-\bu_*(\bx^\pm,t)(t-s)$ are the two inverse Lagrangian images at time $s<t$ 
of $\bx$ at time $t,$ so that ${\mathcal A}_{s;\bx,t} =\{ \ba_-(s;\bx,t),\ba_+(s;\bx,t)\}.$ In particular, if $\bx_\nu$ 
satisfies $\bu_\nu(\bx_\nu,t)=\bar{\bu}_*(\bx,t),$ then $p=1/2$ and 
\be P_\nu^{s;\bx_\nu,t} \stackrel{w}{\longrightarrow} P_*^{s;\bx,t}=\frac{1}{2}\ \delta_{\ba_+(s;\bx,t)} + 
\frac{1}{2} \delta_{\ba_-(s;\bx,t)}. \lb{symm-lim} \ee  
\end{proposition}\lb{spontstoch} 

\begin{proof}
The solutions $\phi_\nu$ have limits $\phi_*$ as $\nu\rightarrow 0$ given by the Lax-Oleinik formula
for the zero-viscosity Burgers solution \cite{BecKhanin07}:  
$$ \phi_*(\bx,t) =\inf_\ba \left[ \frac{|\bx-\ba|^2}{2(t-s)}+\phi_*(\ba,s)\right].$$
This implies existence of the continuous limiting function 
\be S_*(\ba,s|\bx,t) = \lim_{\nu\rightarrow 0} S_\nu(\ba,s|\bx,t) 
                               = \frac{|\bx-\ba|^2}{2(t-s)} + \phi_*(\ba,s)-\phi_*(\bx,t) \ee        
with the properties $S_*(\ba,s|\bx,t)\geq 0$ and $=0$ only for the finite set $\mathcal{A}_{s;\bx,t}\subset \RR^d$ of 
$\ba$-values at which 
the infinimum in the Lax-Oleinik formula is achieved. 
Because velocities $\bu_\nu=\grad\phi_\nu$ are bounded in the limit as $\boldmath \nu\rightarrow 0,$ 
then $S_\nu(\ba,s|\bx_\nu,t) = S_\nu(\ba,s|\bx,t) + O(\nu)$ if $\bx_\nu=\bx+O(\nu),$ 
so that $$\frac{1}{\nu}S_\nu(\ba,s|\bx_\nu,t) = \frac{1}{\nu}S_\nu(\ba,s|\bx,t) + O(1). $$
It follows that outside the finite set $\mathcal{A}_{s;\bx,t}$, 
probabilities for $P_\nu^{s;\bx_\nu,t}$ decay exponentially as $\nu\rightarrow 0.$ These measures are thus 
exponentially tight and have weak limits along subsequences which are supported on atoms in the set 
$\mathcal{A}_{s;\bx,t}.$ 

In the case where $(\bx,t)$ is a regular point of $\bu_*,$ the set $\mathcal{A}_{s;\bx,t}=\{\ba(s;\bx,t)\},$
a singleton, and every weak subsequential limit is the delta measure $\delta_{\ba(s;\bx,t)}.$ 

In the case where $(\bx,t)$ is a generic point in the shock set of $\bu_*,$ the set $\mathcal{A}_{s;\bx,t}=\{\ba_+(s;\bx,t),\, \ba_-(s;\bx,t)\}.$
Hence, every weak subsequential limit is of the form 
$$ p_*\ \delta_{\ba_+(s;\bx,t)} + (1-p_*) \delta_{\ba_-(s;\bx,t)} $$
for some $p_*\in [0,1].$ However, taking the limit as $\nu\rightarrow 0$ of the C-I representation 
$\bu_\nu(\bx_\nu,t) = \int \bu_\nu(\ba,s) P_{\nu}^{s;\bx_\nu,t}(d\ba)$ gives
$$ p \ \bu_*(\bx^+,t) + (1-p) \bu_*(\bx^-,t) =  \int \bu_*(\ba,s) P_{*}^{s;\bx_\nu,t}(d\ba) = 
p_*\ \bu_*(\bx^+,t) + (1-p_*) \bu_*(\bx^-,t), $$ 
or $(p-p_*)(\bu_*(\bx^+,t) -\bu_*(\bx^-,t))=\bzed.$ Since $\bu_*(\bx^+,t)\neq \bu_*(\bx^-,t),$ $p_*=p$ for every subsequence
$\nu_k\rightarrow 0$ and thus (\ref{2delta}) holds.   \hfill $\Box$
\end{proof}

\vspace{20pt}
\noindent {\bf Remark \# 1:} 
Following the approach of \cite{LaForgueOMalley95}, we may define the {\it shock surface} for $\nu>0$ as 
$$ {\mathcal S}_\nu(t) = \left\{ \bx: \ |\bx-\bx_*|= O\left(\frac{\nu}{|\Delta \bu|}\right) \mbox{ for some } \bx_*\in {\mathcal S}_*(t) \mbox{ and } 
\bu_\nu(\bx,t) = \frac{\bu_*(\bx^-_*,t)+\bu_*(\bx^+_*,t)}{2} \right\} $$
where ${\mathcal S}_*(t)$ is the shock surface of the inviscid limit $\bu_*$ and where $\Delta \bu=\bu_--\bu_+.$ 
The previous proposition thus implies that (\ref{symm-lim}) holds for a 
sequence $\bx_\nu\in {\mathcal S}_\nu(t)$ such that $\bx_\nu\rightarrow \bx\in {\mathcal S}_*(t).$
This means that stochastic particles which are ``exactly on the shock'' at time $t$ for $\nu>0$ 
must jump off the shock backward in time as $\nu\rightarrow 0,$ with equal probability to the left or to the right. 

\subsection{One-Dimensional Case}\lb{sec:1D} 

There is special interest in the one-dimensional case, because we wish to use Burgers as a toy model to understand 
the relation between spontaneous stochasticity and anomalous dissipation. It is only in 1D that the local integrals 
$I_\psi(t)=\int \rmd x \ \psi(u(x,t))$ are invariants of smooth inviscid solutions. Furthermore, there is more detailed analysis 
of shock solutions available in 1D which we can exploit. For the remainder of this article we discuss primarily the 
one-dimensional problem.

For a solution $u_\nu$ of 1D Burgers the shock surface ${\mathcal S}_\nu(t)$ 
consists of isolated points $x_\nu^i(t),$ $i=1,2,3,...$.This follows from a matched asymptotic 
analysis of \cite{GoodmanXin92} for systems of conservation laws in 1D. More generally, if $(x,t)$ is a shock point of the 
inviscid Burgers solution $u_*$ where $u_\pm = u_*(x^\pm,t),$ then for any $\upsilon \in (u_+,u_-)$
there exists a unique point $x_\nu(t;\upsilon)$ such that 
$$ u_\nu(x_\nu(t;\upsilon),t) = \upsilon, \,\,\,\, x_\nu(t;\upsilon) = x_*(t) + O(\nu/\Delta u) $$
with $\Delta u=u_--u_+$. Defining the stretched spatial variable $\xi=(x-x_\nu(t;\upsilon))/\nu$ 
similarly as in \cite{LaForgueOMalley95} and thus
$$ \bar{u}_{\nu}(\xi,t;\upsilon) = u_\nu(x_\nu(t;\upsilon)+\nu\xi,t), $$
the 1D viscous Burgers equation becomes 
\be \dot{x}_\nu(t) \bar{u}_\xi -\bar{u}\bar{u}_\xi + \bar{u}_{\xi\xi} = \nu \bar{u}_t. \lb{1dBurg_stretch} \ee
Integrating over $\xi,$ one finds 
$$ \dot{x}_\nu(t;\upsilon) = \frac{1}{2}(u_-+u_+) + \frac{\nu}{u_--u_+} \int_{-\infty}^{+\infty} \rmd \xi\ \bar{u}_t(\xi,t;\upsilon) 
= \frac{1}{2}(u_-+u_+) + O\left(\frac{\nu}{T\Delta u}\right) $$
up to exponentially small terms.  Thus, the points $x_\nu(t;\upsilon)$ for all $\upsilon \in (u_+,u_-)$
``move with the shock" asymptotically for $\nu\rightarrow 0.$ Neglecting the $O(\nu)$ term on the right side 
of eq.(\ref{1dBurg_stretch}), likewise using the similarity variable $\xi=(x-x_*(t))/\nu,$ and therefore 
replacing $\dot{x}_\nu(t)$ with $\dot{x}_*(t)=\frac{1}{2}(u_-+u_+)$ , eq.(\ref{1dBurg_stretch}) becomes the equation for the 
zeroth-order inner solution in the matched asymptotic analysis of \cite{GoodmanXin92}, specialized to Burgers. 
The zeroth-order solution for the boundary conditions $\bar{u}(-\infty,t)=u_-(t),$ $\bar{u}(+\infty,t)=u_+(t)$ is 
$$ \bar{u}_0(\xi,t) = \dot{x}_*(t)-\frac{\Delta u}{2}\tanh\left(\frac{(\xi-\delta_0)\Delta u}{4}\right), $$
where $\delta_0$ is a constant of translation which is undetermined at this order.
To fix this constant, \cite{GoodmanXin92} showed that one must match velocity-gradients 
at next order in the asymptotic expansion. 

It follows for $\upsilon = \frac{1}{2}(1-\lambda)u_-+\frac{1}{2}(1+\lambda)u_+$
with $\lambda\in (-1,1)$ that 
$$  x_\nu(t;\upsilon) \doteq x_*(t) + \nu \left( \frac{4}{\Delta u}\tanh^{-1}(\lambda) + \delta_0\right) + O(\nu^2). $$
These are just the first terms in an asymptotic expansion in powers of $\nu$ that follows from the method
of \cite{GoodmanXin92}. For $\lambda=0$ this expansion gives the shock location in the sense of 
\cite{LaForgueOMalley95}.  Proposition 4.4
implies that the stochastic trajectories moving backward in time from the point $x_\nu(t;\upsilon)$ remain random 
as $\nu\rightarrow 0,$ moving to the left with probability $\frac{1-\lambda}{2}$ and to the right with 
probability $\frac{1+\lambda}{2}$. 
 
It is helpful to illustrate these results by a concrete example, the {\it Khokhlov sawtooth solution} 
of viscous Burgers. This is the velocity field 
\be  u_\nu(x,t) = \frac{x-L\tanh(Lx/2\nu t)}{t} \lb{khokhsol} \ee 
defined for $x\in \RR$ and $t>0$ \cite{SoluyanKhokhlov61}.  
As the previous discussion shows, the Khokhlov solution has the universal form of a 
viscously-smoothed shock in 1D Burgers, sufficiently close to the shock and in its rest frame.  
Since the velocity potential of the Khokhlov solution is
$$\phi_\nu(x,t) = \frac{x^2}{2t}-2\nu\ln \cosh\left(\frac{L x}{2\nu t}\right), $$ 
it is straightforward to calculate from Proposition 4.3
the transition probability 
$$ p_{u_\nu}(a,s|0,t) = \frac{1}{\mathcal{Z}}\cosh\left(\frac{La}{2\nu s}\right) \exp\left(-\frac{t}{4\nu s(t-s)}a^2\right),
\,\,\,\, s<t $$
for this solution. Since $\cosh(x)\sim e^{|x|}/2$ for $|x|\gg 1,$ it follows that the density becomes 
$$ p_{u_\nu}(a,s|0,t) \sim \frac{1}{2\mathcal{Z}}\exp\left[\frac{1}{4\nu s}\left(2L|a|-\frac{t}{t-s}a^2\right)\right], $$
as $\nu\rightarrow 0,$ except for a narrow interval around the origin of width $\sim \nu s/L.$ An application
of the Laplace method shows that the associated probability measure converges to (\ref{2delta}) with
$$ a_\pm(s;0,t) = \pm L \left(1-\frac{s}{t}\right ), $$
which are indeed the Lagrangian pre-images at time $s$ of the shock at the origin at time $t,$ for the limiting velocity 
$u_*(x,t)=\frac{1}{t}[x-L \, {\rm sign}(x)].$ The Laplace method also yields 
$$ \mathcal{Z}^{-1}\sim \sqrt{\frac{t}{4\pi\nu s(t-s)}} \exp\left[-\frac{L^2}{4\nu s}\left(1-\frac{s}{t}\right)\right] $$
for $\nu\rightarrow 0,$ which is the asymptotic value of $p(0,s|0,t).$ Thus, the probability to remain at the 
shock is transcendentally small as $\nu\rightarrow 0,$ for any $s<t.$ This is not surprising, because close 
to the shock center the equation for stochastic Lagrangian trajectories becomes
$$ \rmd \txi = -\left(\frac{L^2}{2\nu t}-1\right)\txi  \frac{\rmd t}{t} +\sqrt{2\nu} \ \rmd \tilde{W}(t) $$
and {\it backward in time} there is a strong repulsion from the shock. This corresponds to a standard 
problem in statistical physics: a noise-induced transition from an unstable equilibrium. It is well-known 
that for weak noise the realizations exit quickly from the unstable point and transit to a 
deterministic solution of the equation without noise \cite{Suzuki78}. 

\subsection{Limiting Backward Process}\lb{sec:LBP}  

Of the limiting probabilities obtained in Proposition 4.4 
there is a distinguished case in which the 
particle starts ``exactly on the shock'' for $\nu>0$ and then jumps off to the right or the left 
with equal probabilities as $\nu\rightarrow 0.$ It is only for this case that the particle drift velocity
at the shock equals the limiting shock velocity. This case corresponds to a random process $\tilde{x}_*(t)$ which 
enjoys the same properties as the processes obtained by the geometric construction in section \ref{sec:geom}.  
It is clearly Markovian backward in time in an extended state space $X(t) \subset \RR \times \{-1,1\}$ with label
$\alpha=+1$ indicating to the right of the shock and $\alpha=-1$ to the left. The Markov generator is 
$$ L(t) f(x,\pm) = -u(x^\pm,t) f'(x,\pm) $$
and initial conditions for $(x,t)$ a shock point of $u_*$ assign $\alpha=+1$ or $-1$ with probability $1/2.$ 
The realizations of this random process move always along straight-line characteristics.  
The Lagrangian velocity $\bar{u}_*(\tilde{x}_*(t),t)$ is also a martingale backward in time. For 
$(x,t)$ a regular point of $u_*,$ this is the usual conservation of velocity along straight-line characteristics, 
while, for $(x,t)$ a shock point of $u_*,$ the martingale property depends also on the definition 
$\bar{u}(x,t)=\frac{1}{2}\left[u(x^+,t)+u(x^-,t)\right].$ Note that, as for the geometric construction,
the martingale property implies in 1D the positivity of dissipation. Indeed, the shock velocity can be represented 
using the martingale property as  
\begin{eqnarray}
u_*(t) &=&  \frac{1}{2} [u_+(t)+u_-(t)] \cr
         &=&  \frac{1}{2} \left[ \left(\frac{x_*(t)-a_+}{t}\right) + \left(\frac{x_*(t)-a_-}{t}\right) \right]  
         = \frac{1}{a_+-a_-}\int_{a_-}^{a_+} \left(\frac{x_*(t)-a}{t}\right) \ \rmd a,
\end{eqnarray}          
which is exactly the condition (\ref{great}) needed to show positivity. The Lax entropy condition $u^->u^+$ 
is implicit in this formulation, since it guarantees that $a^+>a^-.$

\section{Non-uniqueness, Dissipation and a Conjecture}\lb{sec:nonuniq} 

The results of the previous two sections can be restated as follows: the entropy (viscosity, dissipative) solution 
$u$ of inviscid Burgers equation in one space dimension with smooth initial data $u_0$ satisfies the identity 
\be \bar{u}(x,t) = {\mathbb E}[u_0(\tilde{x}(0))|\tilde{x}(t)=x] = \int \rmd a \ u_0(a) p_u(a,0|x,t) \lb{Burgrep} \ee
where ${\mathbb E}$ is expectation with respect to any of the random processes $\tilde{x}(\tau)$ backward in time 
constructed in the previous sections and $p_u(a,s|x,t)$ is the transition probability for this process. The random process 
$\tilde{x}(\tau)$ has the following properties:
\begin{itemize}
\item[(i)] The realizations of the process projected to coordinate space are generalized solutions of the ODE 
$d\tilde{x}/d\tau=\bar{u}(\tilde{x},\tau)$. 
\item[(ii)] The process is Markov backward in time. 
\item[(iii)] The velocity process $D_\tau^+\tilde{x}(\tau)=\bar{u}(\tilde{x}(\tau),\tau)$ is a backward martingale. 
\end{itemize}  
The formula (\ref{Burgrep}) is the analogue of the representation of the weak solutions for passive scalars in the Kraichnan 
model \cite{EvandenEijnden00, LeJanRaimond02} and is an inviscid analogue of the Constantin-Iyer 
representation of viscous Burgers solutions. 
Note that it is a consequence of (i) that, at points of smoothness of $u(x,t),$ the process $\tilde{x}(\tau)$ is deterministic and 
consists of the single characteristic curve which arrives to $(x,t)$. On the other hand, when $(x,t)$ is located on a shock, the 
previous constructions contain an interesting element of non-uniqueness. The two approaches, the geometric one of section 
\ref{sec:geom} and the zero-viscosity limit of section \ref{sec:zerovisc} (which applies also for dimensions $d\geq 1$), lead to quite 
different stochastic processes. Even within the geometric approach there is an important element of non-uniqueness,  
because the time $t_0$ before formation of the first shock
---when the positions are chosen to be uniformly distributed on the Lagrangian interval $[b_-^f,b_+^f]$ --- is completely 
arbitrary. As can be seen from (\ref{backpdf}),(\ref{forepdf}), assuming a uniform distribution on particle positions 
$[b_-^f,b_+^f]$ at time $t_0<t_*$ does not lead to uniform distributions at other times $t<t_*.$  Hence there are uncountably
many distinct definitions of random processes for which (\ref{Burgrep}) is valid and all of the properties (i),(ii),(iii) hold.   

To the non-uniqueness of the stochastic process of backward particle motions there corresponds a similar non-uniqueness
in the Lagrangian expression (\ref{lagdiss3}) of anomalous dissipation. In that expression, one may likewise chose any 
initial time $t_0<t_*$ and represent the dissipation by integrals over the Lagrangian intervals $[b_i^-(t),b_i^+(t)]$ of 
particle positions at time $t_0$ that will have fallen into shocks at time $t$:      
\begin{eqnarray}
&&  \frac{d}{dt} \int_\RR\rmd x \ \psi(u(x,t)) = \frac{1}{t-t_0}
\sum_{i=1}^\infty \int_{b_i^-(t)}^{b_i^+(t)} \Big(\psi(u^*_i(t))-\psi(u_i(t)) - D_\psi^{u_i}(u_i(t),u^*_i(t))\Big)  \ \rmd b, \cr
&& \,\,\,\,\,\,\,\,\,\,\,\,\,\,\,\,\,\,\,\,\,\,\,\,\,\,\,\,\,\,\,\,\,\,\,\,\,\,\,\,\,\,\,\,\,\,\,\,\,\,\,\,\,\,\,\,
u_i(t)=\frac{x_i^*(t)-b}{t-t_0}.  \lb{lagdiss3t0} \end{eqnarray} 
Note that this anomalous dissipation may be directly represented in terms of the corresponding random process 
(with uniform distribution on positions in shock intervals at time $t_0$) as \\
\vspace{-10pt}
\begin{eqnarray} 
&& \frac{d}{dt} \int_\RR\rmd x \ \psi(u(x,t)) = \cr
&& 
\sum_{i=1}^\infty \Delta u_i \, {\mathbb E}\left[\psi(u^*_i(t))-\psi\left(\frac{\tilde{x}(t)-\tilde{x}(t_0)}{t-t_0}\right) 
- D_\psi^{u_i}\left(\frac{\tilde{x}(t)-\tilde{x}(t_0)}{t-t_0}, \,u^*_i(t)\right)\Big|\dot{\tilde{x}}(t)=u_i^*(t)\right],  \,\,\,\,\,\,\,\,\,\,\,
\end{eqnarray} 
where $\Delta u_i=u_i^--u_i^+>0.$
The negative sign of the dissipation is then seen to be directly due to the backward martingale property (iii) of the random 
process. Note indeed that it is a consequence of the Bauer-Bernard \cite{BauerBernard99} definition that solutions of $\dot{x}=\bar{u}(x,t)$ 
in their sense satisfy 
$$ \tilde{x}(t')-\tilde{x}(t'') = \int_{t''}^{t'} \rmd \tau \ \bar{u}(\tilde{x}(\tau),\tau). $$
Therefore, integrating with respect to time $t$ in Proposition 
3.2 gives as a direct corollary 
$$ {\mathbb E}\Big(\frac{\tilde{x}(t')-\tilde{x}(t'')}{t'-t''}\Big|\dot{\tilde{x}}(t)=u\Big) = u \,\,\,\, \mbox{for all $t''<t'\leq t,$ $\,\,\,\, t'',t',t\in [0,t_f].$} $$
This property for $t'=t,$ $t''=t_0$ 
and, in particular, the consequence that for any convex function $\psi$ 
$$ {\mathbb E}\Big[\psi\Big(\frac{\tilde{x}(t)-\tilde{x}(t_0)}{t-t_0}\Big)\Big|\dot{x}(t)=u\Big] \geq \psi(u) \,\,\,\, \mbox{for all $s\in [0,t_f].$} $$
is thus the basic property required to show dissipativity of the Burgers solution. 

The above fundamental connection between the backward martingale property and dissipation motivates the following:

\vspace{10pt}
\noindent {\bf Conjecture:} The only space-time velocity field $u$ on ${\mathbb R}\times [0,t_f]$ which satisfies 
the identity (\ref{Burgrep}) for a stochastic process $\tilde{x}(\tau)$ with the properties (i),(ii),(iii) is the unique 
viscosity (entropy, dissipative) solution of inviscid Burgers with initial condition $u_0$ on ${\mathbb R}.$
\footnote{The results of section \ref{sec:zerovisc} suggest that the conjecture should also hold for space dimensions 
$d\geq 1,$ although the entropy conditions are no longer valid.} 
\vspace{10pt} 

\noindent In particular, it should follow directly from the dissipation implied by the backward martingale 
property that the field defined by the stochastic representation satisfies the conditions to be an 
``admissible weak solution'' (e.g. see \cite{Dafermos05}, Def. 6.2.1). The condition imposed is highly 
implicit, since the velocity $u$ which appears as the result of the average in (\ref{Burgrep}) is the same 
as the velocity $u$ which appears in the ODE in (i) governing particle motion. The conjecture as stated 
above is not explicit enough to be subject to proof or disproof or even to be entirely well-formulated, without 
additional conditions. A natural requirement on spatial regularity is that the velocity field be of bounded-variation 
at each fixed time $t,$ $u(\cdot,t)\in BV({\mathbb R})$
for all $t\in [0,t_f].$ In that case $u(\cdot,t)$ is continuous except at a countable set of points where right-
and left-hand limits exist, so that the field $\bar{u}(x,t)=\frac{1}{2}(u(x-,t)+u(x+,t)$ is well-defined. Some temporal 
regularity must also certainly be assumed, such as $u\in C([0,t_f];L^1({\mathbb R})).$

It is important to emphasize that the uniqueness claim in the conjecture above is for the weak solution
$u$ only and not for the random process $\tilde{x}.$ As we have already seen by explicit construction, there 
is more than one such random process $\tilde{x}$ for the same entropy solution $u.$ There is thus an arbitrariness 
in how the Burgers velocities can be regarded to be transported by their own flow. This is similar to the arbitrariness 
that exists even for some smooth problems, e.g. the Lie-transport of a magnetic field ${\bf B}$ (closed 2-form) by a smooth 
velocity field $\bu,$ governed by the induction equation 
$$ \partial_t{\bf B} = \grad\times (\bu\times{\bf B}). $$
It has long been known (e.g. \cite{Newcomb58}) that there is more than one ``motion'' which can 
be consistently ascribed to the magnetic field-lines governed by the above equation. This arbitrariness 
holds in that case even for the linear problem of passive transport of the magnetic field by a smooth 
velocity. We shall next explore such passive transport problems for 1D Burgers, which may be a toy 
model of such transport in more realistic situations, as previously considered in \cite{Woyczynski98,BauerBernard99}.
In one dimension there are two types of geometric transport which are possible, passive densities (1-forms) and 
passive scalars (0-forms). We consider these in the following two sections. 

\section{Passive Densities}\lb{sec:density} 

A density ($d$-form) is transported according to the continuity equation, which for $d=1$ becomes
\be  \rho_t + (u\rho)_x =0. \lb{rhoeq} \ee
For smooth solutions there is the explicit solution 
\be \rho(x,t) = \left.\frac{\rho(a,t_0)}{|\xi'_{t_0,t}(a)|}\right|_{\alpha_{t_0,t}(x)} =\int \rmd a \ \delta(x-\xi_{t_0,t}(a)) \rho(a,t_0) 
\lb{rhosol} \ee
given by the flow maps $\xi_{t_0,t}$ generated by $u$ and the inverse maps $\alpha_{t_0,t}=\xi_{t_0,t}^{-1}.$
In integrated form, with $M([x,x'],t)=\int_x^{x'} \rmd y \ \rho(y,t)$ the mass in the interval $[x,x'],$ (\ref{rhosol}) 
expresses mass conservation:  
$$ M([x,x'],t) = M([\alpha_{t_0,t}(x),\alpha_{t_0,t}(x')],t_0). $$
This is the Lie-derivative Theorem for $\rho$ transported by $u$ as a differential 1-form

For non-smooth fields there may, of course, be many distinct weak solutions of (\ref{rhoeq}) obtained 
by different regularizations and limits. For compressible Navier-Stokes fluids the equation (\ref{rhoeq}) is 
unchanged but implicitly regularized by the addition of viscous terms to the momentum equation, which 
smooth the velocity $u.$  Instead, \cite{GawedzkiVergassola00,BauerBernard99} 
have explicitly regularized (\ref{rhoeq}) with a positive diffusivity $\kappa>0$
as well as possibly a smoothed velocity $u_\nu:$
\be \rho_t+(u_\nu\rho)_x=\kappa\rho_{xx}. \lb{rhoeqkap} \ee
In molecular fluids the mass density is protected from any such dissipative transport, because its space 
flux vector is itself a conserved density (momentum density, ${\bf j}=\rho\bu$). The regularization 
(\ref{rhoeqkap}) has however some appealing mathematical properties. For example, the solution 
(\ref{rhosol}) can be generalized by means of the stochastic flows solving the forward It$\bar{{\rm o}}$ SDE
\be d\tilde{\xi}_{t_0,t}(a)=u_\nu(\tilde{\xi}_{t_0,t}(a),t)dt +\sqrt{2\kappa}\,\rmd\tilde{W}(t), \,\,\tilde{\xi}_{t_0,t_0}(a)=a. 
\lb{fSDE} \ee
The first formula in (\ref{rhosol}) is replaced by  
$$ \rho(x,t) = \mathbb{E}\left[ \left.\frac{\rho(a,t_0)}{|\tilde{\xi}'_{t_0,t}(a)|}\right|_{\tilde{\alpha}_{t_0,t}(x)}\right],
\,\,\,\, t>t_0. $$ 
where $\mathbb{E}$ is the expectation over the Brownian motion in (\ref{fSDE}). 
See \cite{GawedzkiVergassola00}. In this formulation, the mass in an interval becomes 
a backward martingale, so that 
$$ M([x,x'],t) = \mathbb{E}\left[M([\tilde{\alpha}_{t_0,t}(x),\tilde{\alpha}_{t_0,t}(x'),t_0)\right]. $$ 
The second formula in (\ref{rhosol}) is also generalized, as  
\be \rho(x,t)=\int \rmd a \ p_{\nu,\kappa}(x,t|a,t_0) \rho_0(a), \,\,\,\, t>t_0 \lb{rhosolkap} \ee
where $p_{\nu,\kappa}(x,t|a,t_0)=\mathbb{E}\left[ \delta(x-\tilde{\xi}_{t_0,t}(a))\right]$ is the transition 
probability for the forward diffusion process.  

It turns out that the weak solutions obtained from the limit $\nu,\kappa\rightarrow 0$ of eq.(\ref{rhoeqkap}) 
have also great physical interest for Burgers. The work of \cite{KhaninSobolevski10,KhaninSobolevski12} 
shows that for two choices of Prandtl number $Pr\equiv\nu/\kappa,$ $Pr=0$ and $Pr=\infty,$ the forward stochastic flow 
$\tilde{\xi}_{t_0,t}(a)$ converges in the limits $\kappa\rightarrow 0$ and $\nu\rightarrow 0,$ respectively, to 
the deterministic, forward coalescing flow for Burgers. In that case, $p_{\nu,\kappa}(x,t|a,t_0)$ converges to 
$$   p^*(x,t|a,t_0) = \left\{ \begin{array}{ll}
                                    \delta(x-\xi_{t_0,t}(a)) &  \mbox{$u$ smooth at $x$} \cr
                                    \delta(x-x_i^*(t)) \chi_{[a_i^-(t),a_i^+(t)]}(a) & \mbox{$u$ with shock at $x=x_i^*(t)$} \cr
                                    \end{array} \right. $$
Although we know of no rigorous proof, it is plausible that the same limit is obtained as $\nu,\kappa\rightarrow 0$ 
for any fixed value of $Pr.$  In that case, the limit of (\ref{rhosolkap}) gives                                 
$$ \rho(x,t) = \rho_0(\alpha_{t_0,t}(x))|\alpha_{t_0,t}'(x)|+\sum_i \delta(x-x_i^*(t))M_i(t), $$
where the first smooth part is well-defined except at shock points $(x,t)$ and the delta function part contains 
$$ M_i(t)=\int_{a_i^-(t)}^{a_i^+(t)} \, \rmd a \ \rho_0(a), $$
which is the mass absorbed into the $i$th shock at time $t>t_0.$ The measure $\rho(\cdot,t)$ is well-defined for initial 
density $\rho_0$ any positive Radon measure. Using the elementary result 
$$ \dot{x}^*_i(t) = [1+(t-t_0)u_0'(a_i^\pm(t))] \dot{a}_i^\pm(t) + u^\pm_i(t), $$
it is straightforward to check that this measure gives a  weak 
(distributional) solution of the Burgers-mass transport system: 
$$ u_t+(\frac{1}{2}u^2)_x=0, \,\,\,\, \rho_t+(u\rho)_x=0. $$
Because $\rho$ is a measure with atoms at the shock, it is essential here to use the convention that 
$u=\bar{u}$ at shock points.  This weak solution corresponds to the ``adhesion model'' widely employed in studies 
of the cosmological mass distribution \cite{GurbatovSaichev84,Vergassolaetal94}. 

As was pointed out by Brenier \& Grenier \cite{BrenierGrenier98},  this weak solution does {\it not} satisfy 
momentum conservation. It is thus also true that the 
entropy inequalities $(\rho h(u))_t+(u\rho h(u))_x\leq 0$ do {\it not} hold for all convex $h,$ since 
this would imply momentum conservation. 
Instead there is a momentum-conservation anomaly which is explicitly calculable for piecewise smooth solutions, as:  
\begin{eqnarray*}
&& (\rho u)_t + (\rho u^2)_x = \sum_i \delta(x-x_i^*(t)) \Big\{ \frac{d}{dt}\left[M_i(t)\dot{x}_i^*(t)\right] \cr
&& \hspace{60pt}
       - \left[ (\rho u)_i^-(t)(u_i^-(t)-\dot{x}_i^*(t)) - (\rho u)_i^+(t)(u_i^+(t)-\dot{x}_i^*(t)) \right]\Big\} \cr
&& \,\,\,\,\,\,\,\,\,\,\,\,
= -\frac{1}{4}\sum_i \delta(x-x_i^*(t)) \left[ (\Delta u_i)^2(\rho_i^- -\rho_i^+) + \Delta u_i (u_i^{-\prime}-u_i^{+\prime})M_i(t)\right]. 
\end{eqnarray*}        
Here, $u_i^{\pm\prime}(t)=u_x(x_i^*(t)\pm,t)$. The first expression for the anomaly is obtained by a standard elementary 
calculation. It has a simple physical meaning, since the term inside the curly bracket is the rate of change of momentum 
of the shock minus the flux of momentum from the left and the right into the moving shock. The second expression 
for the anomaly can be obtained from the first using the easily obtained relations 
$$ \dot{M}_i(t)=\frac{1}{2}\Delta u_i (\rho_i^-+\rho_i^-), \,\,\,\,\,\, \ddot{x}_i^*(t) =- \frac{1}{4} \Delta u_i (u_i^{-\prime}-u_i^{+\prime}). $$
(It follows, incidentally, that $\dot{M}_i\geq 0.$) This anomaly may have either sign, as can be seen from the Khokhlov sawtooth example.
If one starts with initial density $\rho_0(a)=\rho_0^\pm$ for ${\rm sign}(a)=\pm 1,$ it evolves in the Khokhlov flow to   
$\rho(x,t)=\rho_0^\pm \cdot t_0/t$ for ${\rm sign}(x)=\pm 1$ and $t>t_0.$ The momentum anomaly from the shock at the origin 
in the Khokhlov solution is explicitly $(\rho^+(t)-\rho^-(t))\left(\frac{L}{t}\right)^2=(\rho^+_0-\rho^-_0)\left(\frac{t_0}{t}\right)\left(\frac{L}{t}\right)^2.$  
Its sign is determined by the relative magnitude of $\rho^+_0$ and $\rho^-_0,$ which can be arbitrarily chosen.

\section{Passive Scalars}\lb{sec:scalar} 

The other transport problem of interest is advection of a passive scalar (0-form), governed by the equation
\be \theta_t + u \theta_x = 0 . \lb{thetaeq} \ee 
When $u,\theta$ are smooth, then pointwise scalar values are ``frozen-in'' and conserved along Lagrangian trajectories, 
so that 
\be \theta(x,t) = \theta_0(\xi_{t,t_0}(x)) =  \int \rmd a \ \delta(a-\xi_{t,t_0}(x)) \theta_0(a). \lb{scal1pt} \ee
The physically natural regularizations of (\ref{thetaeq}) are a smoothed velocity $u_\nu$ and a molecular
diffusivity $\kappa>0,$ so that $\theta$ satisfies
\be \theta_t + u_\nu \theta_x = \kappa\theta_{xx}. \lb{thetaeqreg} \ee 
In that case, (\ref{scal1pt}) is generalized to 
\be \theta(x,t)=\int \rmd a \ p_{\nu,\kappa}(a,t_0|x,t) \theta_0(a), \lb{scal1ptreg} \ee
where $p_{\nu,\kappa}(a,s|x,t)$ is the transition probability for the backward diffusion:
\be d\tilde{\xi}_{t,s}(x)=u_\nu(\tilde{\xi}_{t,s}(x),s)ds +\sqrt{2\kappa}\,\hd\tilde{W}(s), s<t; \,\,\,\,\tilde{\xi}_{t,t}(x)=x. 
\lb{bSDE} \ee
This is the same formula used to analyze passive scalar advection in the Kraichnan model \cite{Falkovichetal01}.
The nature of the solutions of (\ref{thetaeqreg}) depends on the 
behavior of the statistics of the backward-in-time diffusion process (\ref{bSDE}).   

We shall show that there is {\it spontaneous stochasticity} backward-in-time in (\ref{bSDE}) for $u_\nu$ a Burgers solutions 
with shocks, at any finite value of $Pr=\nu/\kappa.$ That is, the statistics of (\ref{bSDE}) remain random as 
$\nu,\kappa\rightarrow 0$ with $Pr$ fixed. We do not have a general proof which covers an entire class of shock solutions, 
as for the $Pr=1$ case in section \ref{sec:spontstoch}, but we shall prove the 
result for the Khokhlov sawtooth shock which, we have seen, exemplifies the universal form of the viscous Burgers
shock. For simplicity we present here the proof for an even more elementary case, a stationary shock solution of Burgers equation,
$$ u_\nu(x) =-u_0 \tanh(u_0x/2\nu), $$
where, obviously, $u_0=\Delta u/2.$ The details for the Khokhlov solution are very similar but somewhat more complicated, and are presented 
in Appendix B. 
To treat evolution backward in time in the most transparent manner, 
we take $u_\nu(x)\rightarrow -u_\nu(x)$ and discuss instead a stochastic particle position $\txi(t)$ starting at $\txi(0)=0$ 
and moving forward for $t>0$.  Consider then 
\be \rmd\txi(t)= u_0 \tanh(u_0\txi/2\nu) \rmd t +\sqrt{2\kappa}\,\rmd\tilde{W}(t), \,\,\,\, \txi(0)=0.  \lb{txieq} \ee
We prove the following: 

 \begin{proposition}
 For solution $\txi(t)$ of (\ref{txieq}) with $\txi(0)=0$ at any $t>0,$ for $\alpha=\min\{1,\frac{1}{Pr}\},$ and for any $\epsilon\in (0,1),$
 $$ \lim_{\stackrel{\kappa\rightarrow 0}{Pr=\nu/\kappa \,\,\,\, {\rm fixed}}} P_{\nu,\kappa}\Big( |\txi(t)| \geq 
 \alpha (1-\epsilon) u_0 t\Big) =1. $$
 \end{proposition} 
 
\begin{proof}
Consider the velocity potential 
$$ \phi_\nu(x) = 2\nu \ln \cosh(u_0 x/2\nu) $$ 
which is non-negative, less than $u_0|x|,$ and convergent to $u_0|x|$ as $\nu\rightarrow 0$. By the (forward) 
It$\bar{{\rm o}}$ formula
$$ \rmd\phi_\nu(\txi(t),t) = u_0^2 [\tanh^2(u_0\txi/2\nu) + {\rm sech}^2(u_0\txi/2\nu)/Pr] \rmd t 
+ \sqrt{2\kappa}\ u_0  \tanh(u_0\txi/2\nu) \rmd \tilde{W}(t). $$
It follows using $\tanh^2(z) + {\rm sech}^2(z)=1$ that  
$$ \phi_\nu(\txi(t)) >  \alpha u_0^2 t + \sqrt{2\kappa} \ u_0  \int_0^t \tanh(u_0 \txi(s)/2\nu) \rmd \tilde{W}(s) $$
with $\alpha=\min\left\{1,\frac{1}{Pr}\right\}$. Hence for any $\epsilon\in (0,1)$
\begin{eqnarray*}
P_{\nu,\kappa}\bigg( \phi_\nu(\txi(t)) < (1-\epsilon)\alpha u_0^2 t \bigg) &<& 
P_{\nu,\kappa}\left( \sqrt{2\kappa}\int_0^t  \tanh(u_0\txi(s)/2\nu) \rmd \tilde{W}(s) <-\epsilon \alpha u_0 t  \right)\cr
&=& \frac{1}{2} 
P_{\nu,\kappa}\left( \left| \sqrt{2\kappa}\int_0^t  \tanh(u_0\txi(s)/2\nu) \rmd \tilde{W}(s) \right| > \epsilon \alpha u_0 t  \right)
\end{eqnarray*} 
By It$\bar{\rm o}$ isometry the variance of $\tilde{\eta}(t)=\sqrt{2\kappa}\int_0^t  \tanh(u_0\txi(s)/2\nu) \rmd \tilde{W}(s)$ is
$$     {\mathbb E}_{\nu,\kappa}(\tilde{\eta}^2(t)) = 2\kappa \int_0^t  {\mathbb E}_{\nu,\kappa}\big(\tanh^2(u_0\txi(s)/2\nu)\big) ds < 2 \kappa t. $$
Thus the Chebyshev inequality 
$$ P_{\nu,\kappa}\left( \left| \sqrt{2\kappa}\int_0^t  \tanh(u_0\txi(s)/2\nu) \rmd \tilde{W}(s) \right| > \epsilon \alpha u_0 t  \right)
<  2\kappa / \epsilon^2\alpha^2 u_0^2 t, $$
gives 
\be P_{\nu,\kappa}\big(|\txi(t)| < (1-\epsilon)\alpha u_0 t \big) < P_{\nu,\kappa}\big(\phi_\nu(\txi(t)) < (1-\epsilon)\alpha u_0^2 t \big) 
< \kappa / \epsilon^2 \alpha^2 u_0^2 t, \lb{Pbd} \ee
which completes the proof. \hfill $\Box$
\end{proof}
\noindent This result shows that the particle does not remain at its initial position $\txi(0)=0,$ but instead moves 
away from the origin at least at speed $\alpha u_0 $ as $\nu,\kappa\rightarrow 0.$ As a matter of fact, the speed 
must be $u_0$ even when $\alpha<1$. Standard theorems on zero-noise limits for smooth dynamics show that 
the motion becomes deterministic with constant speed $u_0$ away from the origin \cite{FreidlinWentzell98}. 
By symmetry, the probabilities for the particle to move right or left must be equal. We thus obtain two limiting particle 
trajectories $\xi_+(t)$ and $\xi_-(t)$, right moving and left moving at speed $u_0$, with probabilities $1/2.$           
              
\vspace{20pt}
\noindent {\bf Remark \#1:} The previous proof is valid even for $Pr=0$ ($\nu=0$) directly, when  
$ u_*(x) =-u_0 {\rm sign}(x) $ and, after reversal $u_*\rightarrow -u_*,$ $\phi_*(x) = u_0|x|.$
In that case,
$$ \rmd\txi= u_0 {\rm sign}(\txi) \rmd t +\sqrt{2\kappa}\,\rmd\tilde{W}(t), \,\,\,\, \txi(0)=0 $$
leads to 
$$ \rmd |\txi(t)|= u_0 dt + \kappa\, \rmd\tilde{L}(t) + \sqrt{2\kappa}\, u_0 \, {\rm sign}(\txi) \rmd\tilde{W}(t). $$
Here $\tilde{L}(t) = \int_0^t \delta (\tilde{W}(s))$ is the local time process of Brownian motion at 0
(see \cite{KaratasShreve91}), 
which replaces the hyperbolic secant square term in the equation for $\rmd\phi_\nu$ above. Since 
${\tilde L}(t)\geq 0$ a.s., the previous proof goes through unchanged, with $\alpha=1.$ 

The case $Pr=\infty$ ($\kappa=0$) is much more delicate. As for 
the Kraichnan model, any spontaneous stochasticity must now arise from randomness in the initial  
data $\xi(0)$ for the deterministic ODE $\dot{\xi}=u_\nu(\xi)$ rather than from random noise \cite{EvandenEijnden00}. 
It is clear that particle trajectories will remain stochastic for suitable random initial data (or, really, final data, since the 
evolution considered is backward in time) which become deterministic as $\nu\rightarrow 0$. For example, if the data 
are spread continuously over a symmetric interval of length $\sim \nu/u_0$ about 0, then the initial velocities are spread
over an interval $\sim (-u_0,u_0).$ It is obvious that trajectories will then escape with positive probabilities to both 
right and left of the shock. We do not offer here any more precise statement. See the recent preprint of 
Frishman-Falkovich \cite{FrishmanFalkovich14} for more specific analysis of the $Pr=\infty$ problem.

\vspace{20pt}
\noindent {\bf Remark \#2:} The presence of the factor $\alpha$ in the previous rigorous argument suggests that 
escape from the origin may be retarded for $Pr>1$, when $\alpha<1.$ One can heuristically estimate  the 
escape time as $\tau_{esc}\sim \alpha^{-1} \frac{\kappa}{u_0^2}=\frac{\max\{\kappa,\nu\}}{u_0^2}$ and the 
distance required to escape as $\ell_{esc}= \alpha^{-1/2} \frac{\kappa}{u_0}.$

Consider first $Pr<1$, when $\alpha=1.$ Then at distance $\ell_{esc}=\kappa/u_0,$ $u_\nu(\ell_{esc})=u_0\tanh(\frac{1}{2} Pr^{-1})$
$\sim u_0.$ 
Hence, a particle starting at distance $\ell_{esc}$ would diffuse back to the origin in the time $\ell_{esc}^2/\kappa\sim \kappa/u_0^2\sim 
\tau_{esc},$ 
which is the same as the time $\ell_{esc}/u_\nu(\ell_{esc}) \sim \tau_{esc}$ required to move distance $\ell_{esc}$ 
away from the origin by advection. Hence particles further away from the origin than $\ell_{esc}=\kappa/u_0$ 
are unlikely to return. This heuristic   
argument is in agreement with the previous rigorous argument.  Define 
$$ t_B = B \frac{\kappa}{u_0^2}=B \tau_{esc} . $$
for some large constant $B\gg 1. $ The probability bound (\ref{Pbd}) for $\epsilon=1/2$ and $\alpha=1$ becomes 
$$ P_{\nu,\kappa}\big(|\txi(t_B)| < B \ell_{esc} \big) < \frac{4}{B}. $$
We thus see that $\tau_{esc}$ is the characteristic time to escape from the 
shock backward in time. 

Next consider $Pr>1.$ A similar argument suggests that the escape distance is $\ell_{esc}= Pr^{1/2} \kappa/u_0.$
In fact, at that distance $u_\nu(\ell_{esc}) = u_0\tanh(\frac{1}{2}Pr^{-1/2}) \sim u_0 Pr^{-1/2}.$ Thus, in time $\tau_{esc}\sim Pr \cdot \kappa/u_0^2$
the distance moved by diffusion $, (\kappa \tau_{esc})^{1/2},$ and the distance moved by advection, $u_\nu(\ell_{esc})\tau_{esc},$
both equal $\ell_{esc}$. At smaller distances diffusion dominates and at larger distances advection away from the 
origin dominates, supporting the idea that $\ell_{esc}= Pr^{1/2} \kappa/u_0$ is the escape distance and 
$\tau_{esc}\sim Pr \cdot \kappa/u_0^2$ the escape time. This conclusion seems plausible, although the bound 
(\ref{Pbd}) is not sharp enough for $Pr>1$ to verify it.  
 
\vspace{20pt}
\noindent {\bf Remark \#3:} Spontaneous stochasticity at $Pr=0$ is analogous to a {\it zero-temperature 
phase transition} for a one-dimensional spin system in infinite volume. This is true not only for Burgers 
velocities, but in general. 

Consider, for example, the 1-dimensional Ising model in finite-volume $[-N,...,N]$
$$ P_N[\sigma]=\frac{1}{Z} \exp\left(-\frac{1}{k_BT} H[\sigma]\right),\,\,\,\,\,\, H[\sigma]= \frac{J}{2}\sum_{{\tiny \begin{array}{l}
                                                                                                                                                     i=-N\cr
                                                                                                                                                     \sigma_{-N}=+1
                                                                                                                                                     \end{array}}}^N (\sigma_i-\sigma_{i+1})^2$$
with boundary condition $\sigma_{-N}=+1.$ In the zero-temperature limit
$$ P_N[\sigma] \rightarrow \prod_{i=-N}^N \delta_{\sigma_i,+1} \,\,\,\, {\rm as} \,\,\,\, T\rightarrow 0,$$ 
the unique {\it ground-state} with $\sigma_{-N}=+1.$  This is analogous to the fact that the Lagrangian 
path-integral for a smoothed velocity $\bu_\nu$ in the zero-noise limit $\kappa\rightarrow 0$ satisfies            
$$ P^{\nu,\kappa}_\bu[\bx] \,\, \mathcal{D}\bx = \frac{1}{\mathcal{Z}}
     \left. \exp\left(-\frac{1}{4\kappa}\int_{t_0}^{t_f} d\tau\, |\dot{\bx}(\tau)-\bu^\nu(\bx(\tau),\tau)|^2\right)
      \right|_{\bx(t_0)=\bx_0}\mathcal{D}\bx $$ 
$$ \rightarrow \prod_{\tau\in [t_0,t_f]} \delta^3(\bx(\tau)-\bx_*(\tau)) \,\, \mathcal{D}\bx, $$
with $\bx_*(t)$ the unique solution of $\dot{\bx}=\bu^\nu(\bx,t),\,\,\bx(t_0)=\bx_0.$

On the other hand, consider the 1-dimensional Ising model in infinite-volume ($N\rightarrow\infty$)
$$ P_\infty[\sigma]=\frac{1}{Z} \exp\left(-\frac{1}{k_BT} H[\sigma]\right),\,\,\,\,\,\, H[\sigma]=\frac{J}{2}\sum_i (\sigma_i-\sigma_{i+1})^2. $$
In the zero-temperature limit
$$ P_\infty[\sigma] \rightarrow \frac{1}{2}\prod_i \delta_{\sigma_i,+1} +\frac{1}{2}\prod_i \delta_{\sigma_i,-1} \,\,\,\, {\rm as} \,\,\,\, T\rightarrow 0,$$ 
the {\it symmetric mixture of ground-states}, a zero-temperature phase transition. See \cite{Georgii11} for more 
careful statements. Likewise, in the zero-noise limit with {\it first} $\nu\rightarrow 0,$ {\it then} $\kappa\rightarrow 0$, 
the Lagrangian path-integral
$$ P^{0,\kappa}_\bu[\bx] \,\, \mathcal{D}\bx= \frac{1}{\mathcal{Z}}
    \left.  \exp\left(-\frac{1}{4\kappa}\int_{t_0}^t d\tau\, |\dot{\bx}(\tau)-\bu(\bx(\tau),\tau)|^2\right)
    \right|_{\bx(t_0)=\bx_0}\mathcal{D}\bx \,\, $$ 
$$ \longrightarrow \int \Pi(d\alpha) \prod_{\tau \in [t_0,t]} \delta^3(\bx(\tau)-\bx_\alpha(\tau))\,\, \mathcal{D}\bx. $$
where $\Pi$ is a nontrivial probability measure on the non-unique solutions $x_\alpha$ of $\dot{\bx}=\bu(\bx,t),\,\,\bx(t_0)=\bx_0.$ 
These solutions are the analogues of the zero-temperature ground states\footnote{There is a strong analogy of the high-Reynolds limit 
of turbulence with the semi-classical limit of quantum
mechanics, e.g. see \cite{Kraichnan75}, section 6. This raises the possibility of {\it quantum spontaneous stochasticity}. 
Consider a non-relativistic quantum-mechanical particle of mass $m$ and electric charge $q$ 
moving in an electric field ${\bf E}=-\grad\Phi-(1/c)\partial_t\bA$ 
and magnetic field ${\bf B} =\grad\btimes\bA $ governed by the Schr\"odinger equation 
$$ i\hbar \partial_t\Psi = \frac{1}{2m}\left(-i\hbar\grad-\frac{q}{c}\bA\right)^2\Psi+q\Phi\Psi, $$
or with transition amplitudes given by Feynman's path-integral formula \cite{Feynman48}
$$ \langle \bx,t|\bx_0,0\rangle =\int_{\bx(0)=\bx_0}^{\bx(t)=\bx} {\mathcal D}\bx 
\exp\left(\frac{i}{\hbar}\int_0^t ds\ L(\bx(s),\dot{\bx}(s),s)\right), $$
where the classical Lagrangian is 
$$ L(\bx,\dot{\bx},t)=\frac{1}{2}m|\dot{\bx}|^2+\frac{q}{c}\bA(\bx,t)\cdot\dot{\bx}-q\Phi(\bx,t). $$
When ${\bf E},{\bf B}$ are Lipschitz, then the classical equations of motion 
$$ m\ddot{\bx} =q\left[{\bf E}(\bx,t) + \frac{1}{c}\dot{\bx}\ \btimes\ {\bf B}(\bx,t)\right] $$
have unique solutions and the stationary phase argument of Feynman \cite{Feynman48} 
yields classical dynamics for $\hbar\rightarrow 0.$ However, if the electromagnetic fields are non-Lipschitz,
then quantum superposition effects could persist in the classical limit. This would presumably require
classical electromagnetic fields which are ``rough'' down to the de Broglie wavelength $\lambda=h/p$ of the particle.}.

The comparison with the one-dimensional Ising model is especially apt for Burgers, since we have seen that there are 
precisely two ``ground states'' in the zero-noise limit for Burgers\footnote{As shown in section \ref{sec:geom}, 
there are actually uncountably many solutions $\xi(s)$ of $D_s^+\xi(s)=\bar{u}(\xi(s),s)$ for a final point $(x,t)$
in the shock set of the Burgers solution $u$. Note by Kneser's Theorem \cite{Hartman02} that there are in general 
uncountably many solutions of the initial-value problem for velocity fields $u$ which are continuous, when 
non-uniqueness of solutions occurs at all. However,  only two of these ``ground states'', the extremal solutions, 
are selected for Burgers by the zero-noise limit.}, a trajectory leaving to the right of the shock and another 
leaving to the left, with equal probabilities. One can say, roughly, that Burgers is in the ``Ising class.''  

\vspace{20pt}
\noindent With the above information on the statistics of the backward-in-time diffusion process (\ref{bSDE}) for Burgers, we 
can now discuss the limits $\nu,\kappa\rightarrow 0$ of the solutions of the passive scalar equation (\ref{thetaeqreg}). The 
transtion probabilities for $0\leq Pr<\infty$ converge weakly to 
$$ p_*(a,t_0|x,t) = \left\{ \begin{array}{ll}
                                    \delta(a-\xi_{t,t_0}(x)) &  \mbox{$u$ smooth at $(x,t)$} \cr
                                     \frac{1}{2}\left[ \delta(a-\xi_-(t_0))+\delta(a-\xi_+(t_0))\right] & \mbox{$u$ with shock at $(x,t),$} \cr
                                    \end{array} \right. $$ 
where, of course, $\xi_\pm(t_0)=a_\pm(t),$ i.e. the endpoints of the Lagrangian interval at time $t_0.$ We therefore
obtain the limit of the scalar solutions to be $\theta_*(x,t)=\int \rmd a \ p_*(a,t_0|x,t) \theta_0(a),$ or 
\be  \theta_*(x,t) = \left\{ \begin{array}{ll}
                                    \theta_0(\xi_{t,t_0}(x)) &  \mbox{$u$ smooth at $x$} \cr
                                     \frac{1}{2}\left[ \theta_0(\xi_-(t_0))+\theta_0(\xi_+(t_0))\right] & \mbox{$u$ with shock at $(x,t).$} \cr
                                    \end{array} \right. \lb{thetasolstar} \ee
Pointwise scalar values $\theta(x,t)$ are not ``frozen-in'' deterministically, i.e. are not generally equal to $\theta_0(a)$ 
for $a=\xi_{t,t_0}(x).$ This property does hold in a probabilistic sense, as the formula (\ref{thetasolstar}) for $t_0=s$
can be rewritten as $\theta_*(x,t)=\mathbb{E}_*[\theta_*(\txi_{t,s}(x),s)]$ where $\txi_{t,s}(x)$ is the ensemble 
of stochastic Lagrangian trajectories obtained from the limit $\nu,\kappa\rightarrow 0.$ The generalization of the 
``frozen-in'' property is that $\theta_*(\txi_{t,s}(x),s)$ for $s<t$ is a backward-in-time martingale.   
 
It is important to emphasize that $\theta_*(x,t)$ given by (\ref{thetasolstar}) is {\it not} a weak (distributional) 
solution of the scalar advection equation (\ref{thetaeq}). In fact, the standard notion of weak solution 
is not usually available for (\ref{thetaeq}), since it is not of conservation form for $\nabla\cdot u\neq 0.$ Formally, 
 (\ref{thetaeq}) can be rewritten as 
$$ \int \rmd x \int \rmd t \left[ \psi_t(x,t) + u(x,t)\psi_x(x,t) + \nabla\cdot u(x,t) \psi(x,t) \right]\theta(x,t) = 0, $$
for a smooth test function $\psi,$ but the expression on the left is generally ill-defined for non-smooth fields $\theta$ 
when $\nabla\cdot u$ itself exists only as a distribution. For Burgers the lefthand side is well-defined for $\theta_*$
defined by (\ref{thetasolstar}) but not equal to zero. One can instead regard (\ref{thetasolstar}) as a new notion of 
a ``Lagrangian weak solution,'' which generalizes the method of characteristics for smooth solution fields 
rather than generalizing the Eulerian equations of motion.  

An important question addressed in the prior work of Bauer \& Bernard \cite{BauerBernard99} is whether passive scalars  
in a Burgers flow preserve the invariants of smooth solutions or whether the scalar conservation-laws 
are afflicted with anomalies. It was concluded in \cite{BauerBernard99} that there are no anomalies 
``in the limit $\kappa\rightarrow 0$''\footnote{It was never clearly specified in 
\cite{BauerBernard99} how the joint limits $\nu\rightarrow 0,$ $\kappa\rightarrow 0$ should be taken. Their analysis 
seems to be best justified for the $Pr=\infty$ problem, taking first $\kappa\rightarrow 0$, then $\nu\rightarrow 0$. 
This case remains open.} for passive scalars 
advected by Burgers velocities with shocks, for scalar quantities of the form 
\be  I_\psi(t) = \int \rmd x \ \psi(\theta(x,t)) \lb{I-BB} \ee
given by a smooth function $\psi.$ They argued that shocks occur only on a set of Lebesgue measure zero and, 
therefore, do not alter the conservation properties of the scalar fields, which remain smooth almost everywhere.  
We disagree both with this conclusion and even with the formulation of the problem. 

In the first place, the quantities in (\ref{I-BB})
are not ``invariants" of the scalar field, since even for smooth fields 
\be  \frac{\rmd}{\rmd t} I_\psi(t) = \int \rmd x \ \psi(\theta(x,t)) (\nabla\cdot u(x,t))\neq 0.  \lb{I-BB-dot} \ee 
Therefore, if pumping and damping of this invariant are provided, the scalar will not enter a statistical stationary state 
in which driving and dissipation are balanced against each other. The natural invariants of a passive scalar are 
instead of the form
\be  J_\psi(t) = \int \rmd x \ \rho(x,t) \ \psi(\theta(x,t)),  \lb{I-ED} \ee  
in which $\rho$ is a conserved density, as in the previous section \ref{sec:density}. Since for smooth solutions
$$ (\rho \psi(\theta))_t + (\rho \psi(\theta) u)_x =0, $$
the quantities $J_\psi(t)$ are indeed conserved quantities for such standard smooth solutions.  

The second disagreement with the conclusions of \cite{BauerBernard99} is that the natural 
integral invariants (\ref{I-ED}) of the scalar for smooth dynamics are, as a matter of fact, afflicted 
with anomalies when using the same regularizations considered by those authors. Thus, when the passive
density and passive scalar are both regularized by smoothing the Burgers velocity with viscosity $\nu>0$
and by adding a diffusivity $\kappa>0$, then the solutions for $\rho$ and $\theta$ obtained in the limit $\nu,\kappa\rightarrow 0$
with $Pr<\infty$ fixed, which were characterized in the previous sections, do not preserve the invariants 
(\ref{I-ED}). Instead, there are conservation-law anomalies which are easily calculable, either in Lagrangian
or in Eulerian forms. The Lagrangian expression derived by the approach of section \ref{sec:dissanom}, 
\begin{eqnarray}
&& \int_{\mathbb{R}} \rho(x,t) \psi(\theta(x,t)) \ \rmd x - \int_{\mathbb{R}} \rho_0(a) \psi(\theta_0(a)) \ \rmd a \cr
&& \,\,\,\,\,\,\,\,\,\,\,\,\,\,\,\,\,\,\,\,\,\,\,\,\,\,\,\,\,\,\,\,\,\,\,\,\,\,\,\,\,\,\,\,\,\,\,\,
=\sum_i \int_{a_i^-(t)}^{a_i^+(t)} [\psi(\theta_i^*(t))-\psi(\theta_0(a))] \rho_0(a) \ \rmd a.
\lb{scal-anom} \end{eqnarray} 
shows that the origin of the scalar anomalies is the loss of information about the initial scalar distribution. 
This is similar to the conservation-law anomalies associated to the Burgers velocity, except that the 
scalar anomalies can have either sign\footnote{This is also true for the diffusive violation of the conservation 
laws in the regularized equations. If the density $\rho$ and scalar $\theta$ are both subject to the same 
diffusivity $\kappa>0,$ then it is not hard to show that
$$ (\rho \psi(\theta))_t + [u_\nu \rho\psi(\theta) -\kappa (\rho \psi(\theta))_x]_x = -\kappa \rho \theta''(\theta)\theta_x^2
-2\kappa \rho_x \psi'(\theta) \theta_x. $$
The first ``dissipation'' term on the right is non-positive when the function $\psi$ is convex, but the second term 
is of indeterminate sign. It would be interesting to know whether the $\kappa\rightarrow 0$ limit 
of this dissipation term yields the same result (\ref{scal-anom}) as for the ``Lagrangian weak 
solution.''}.  The corresponding Eulerian form 
of the scalar anomalies is 
\begin{eqnarray*}
&& (\rho \psi(\theta))_t + (\rho u \psi(\theta))_x = \sum_i \delta(x-x_i^*(t)) (u_i^- -u_i^+)\cr
&& \,\,\,\,\,\,\,\,\,\,\,\, \times 
\left[ \frac{1}{2}\Big(\rho_i^+(\psi(\theta_i^*)-\psi(\theta_i^+))+\rho_i^-(\psi(\theta_i^*)-\psi(\theta_i^-))\Big)-
\frac{1}{4}\psi'(\theta_i^*)(\theta_i^{-\prime}-\theta_i^{+\prime})M_i(t)\right]. 
\end{eqnarray*}       
The argument of \cite{BauerBernard99}, that scalar anomalies are absent 
because shocks occur only on a set of zero Lebesgue measure, 
fails among other reasons\footnote{The dissipative anomalies for the kinetic 
energy of an  inviscid Burgers solution arise, of course, entirely from the shocks. Experiments on the multifractal 
structure of real hydrodynamic turbulence indicate that its energy dissipation set has fractal dimension also 
less than the space dimension (three) \cite{MeneveauSreenivasan91, KestenerArneodo03}. 
Thus, it is not unusual in turbulent systems that dissipative anomalies arise from zero-measure sets! 
If $\nu=\kappa$ $(Pr=1)$ and the 
initial value of the scalar is the same as that for the Burgers velocity, $\theta_0=u_0,$ then solutions $\theta$
and $u$ agree for all times. In that case the ``invariants'' (\ref{I-BB}) considered by \cite{BauerBernard99} 
are also not conserved for $\nu,\kappa\rightarrow 0.$} because the density $\rho$ for the limiting solutions develops 
positive mass atoms precisely at these shocks. 

It is worth noting that there is an anomaly even for $\psi(\theta)=\theta:$
\begin{eqnarray*}
&& (\rho \theta)_t + (\rho u \theta)_x \cr
&& \,\,\,\,\,\,\,\,\,\,\,\,
= -\frac{1}{4}\sum_i \delta(x-x_i^*(t)) (u_i^--u_i^+)\left[(\theta_i^- -\theta_i^+)(\rho_i^- -\rho_i^+) + (\theta_i^{-\prime}-\theta_i^{+\prime})M_i(t)\right]. 
\end{eqnarray*}  
Thus, $\varrho=\rho\theta$ is also not a weak (distributional) solution of $\varrho_t+(u\varrho)_x=0.$         

\section{Time-Asymmetry of Particle Stochasticity}\lb{sec:reversal} 

We have shown in the preceding that there is spontaneous stochasticity in the zero-noise limit at finite-$Pr$ for Lagrangian 
particles in a Burgers flow moving {\it backward in time}. On the contrary, the zero-noise limit forward in time at any $Pr$
should lead to a natural coalescing flow for Burgers \cite{BauerBernard99}, as has been proved rigorously at $Pr=0$ 
\cite{KhaninSobolevski10,KhaninSobolevski12}. The Burgers system is thus quite different from the time-reversible Kraichnan 
model, where strongly compressible flows lead to coalescence both forward and backward in time \cite{GawedzkiVergassola00, 
EvandenEijnden01}. Based on studies in the Kraichnan model, the difference between stochastic splitting or sticking of 
particles has been viewed as a consequence of the degree of compressibility of the velocity field, with weakly 
compressible/near-solenoidal  velocities leading to splitting and strongly compressible/near-potential velocities leading 
to sticking \cite{GawedzkiVergassola00, EvandenEijnden00, EvandenEijnden01, LeJanRaimond02, LeJanRaimond04}. 
However, the Burgers equation with a velocity that is pure potential can produce both sticking and splitting, in different directions of time. 

The Burgers system appears in fact to have a remarkable similarity in particle behaviors to incompressible Navier-Stokes
turbulence, even though the Burgers velocity is pure potential and the Navier-Stokes velocity is pure solenoidal. Because the 
Navier-Stokes equation just as viscous Burgers is not time-reversible, it can exhibit distinct particle behaviors forward 
and backward in time. Navier-Stokes turbulence appears to lead to Richardson 2-particle dispersion and, consequently,
stochastic particle splitting, both forward and backward in time. Remarkably, however, the rate of dispersion is found in 
empirical studies of three-dimensional Navier-Stokes turbulence to be greater backward in time than forward 
\cite{Sawfordetal05,Bergetal06,Eyink11}. This is the same tendency seen in a 
very extreme form in Burgers, where there is particle splitting backward in time but only coalescence forward in time. 

We have also shown in this work, at least for Burgers, that there is a direct connection between spontaneous stochasticity 
and anomalous dissipation for hydrodynamic equations, as had been suggested earlier in \cite{GawedzkiVergassola00}. 
More precisely, the relation we have found is between the sign of conservation-law anomalies in Burgers and spontaneous 
stochasticity backward in time. In turbulence language, the {\it direct cascade} of energy to small scales in Burgers is due to 
stochastic particle splitting backward in time. The empirical observations on particle dispersion in Navier-Stokes turbulence 
cited above lead us to suggest more generally {\it a deep relation between cascade direction and the time-asymmetry 
of particle dispersion}. Indeed, three-dimensional Navier-Stokes turbulence has a forward cascade of energy, just as Burgers, 
and likewise a faster particle dispersion backward in time than forward. This conjecture is strengthened by the numerical observation 
of a reversed asymmetry for the inverse energy cascade of two-dimensional turbulence, with Richardson particle dispersion 
in 2D inverse cascade instead faster forward in time than backward \cite{FaberVassilicos09}.   

There is a well-known connection in statistical physics between dissipation/entropy production and the 
asymmetry between forward and backward processes, embodied in so-called {\it fluctuation theorems}. For a 
recent review of this theory, see \cite{Gawedzki13}. Since it is natural to suspect a relation with our 
conjectures above, we briefly recall here that the fluctuation theorems state that 
\be {\mathbb E}(e^{W[\tilde{\bx}]}) =1 \lb{FT} \ee
where $e^{W} = d{\mathcal P}'/d{\mathcal P}$ is a Radon-Nikod\'ym derivative of the path measure 
for a time- reversed process with respect to the path measure for the direct process. Physically, 
$-k_B � W$ has often the meaning of ``entropy production'' and the consequence of Jensen's inequality, 
$$ {\mathbb E}(W[\tilde{\bx}])\leq 0, $$
implies the sign of energy dissipation or entropy production in the 2nd law of thermodynamics. 
However, the fluctuation theorems are a considerable refinement of the 2nd law, since they 
state not only the existence of entropy production on average but also provide information about 
the likelihood of 2nd-law violations.  

Fluctuation theorems are straightforward to derive for stochastic particle motion in Burgers governed by 
the SDE 
$$ d\tilde{\bx} = \bu(\tilde{\bx},t) \rmd t +\sqrt{2\nu}\ \rmd \tilde{\bW}(t), \,\,\,\, t\in [t_0,t_f], $$
especially when the velocity is potential with $\bu(\bx,t)=\grad\phi(\bx,t).$ In this case, 
the time-reverse process is the same as the direct process with merely the time-change $t' = t_0+t_f -t$ 
\cite{Kolmogorov35}. Also, for any gradient dynamics with additive noise, the accompanying measures 
(or instantaneously stationary measures) at time $t$ are
$$ n_t(\rmd \bx) = \frac{1}{Z_t} \exp\left(\frac{\phi(\bx,t)}{\nu}\right). $$
The standard recipes \cite{Gawedzki13} then give (\ref{FT}) with  
\begin{eqnarray*} 
&& W[\tilde{\bx}] =  \frac{1}{\nu}\phi(\tilde{\bx}(t_f),t_f)-\ln \rho_f(\tilde{\bx}(t_f)) \cr 
&& \hspace{50pt} 
+ \frac{1}{\nu}\int_{t_0}^{t_f} \partial_t\phi (\tilde{\bx}(t),t) \rmd t
-\frac{1}{\nu}\phi(\tilde{\bx}(t_0),t_0)+\ln \rho_0(\tilde{\bx}(t_0))
\end{eqnarray*} 
where $\rho_0(\bx)$ and $\rho_f(\bx)$ are starting probability densities for the forward and backward processes
which may be freely chosen.  The trajectories $\tilde{\bx}(t)$ in the expectation ${\mathbb E}$ of (\ref{FT}) 
are sampled from solutions of the forward SDE with initial data chosen from $\rho_0.$ 

An intriguing question is whether such fluctuation theorems for stochastic particle motion in Burgers have any relation 
with anomalous dissipation in the limit $\nu\rightarrow 0.$ For Burgers the potential satisfies the KPZ/Hamilton-Jacobi 
equation
$$ \partial_t\phi_\nu(\tilde{\bx}(t),t) = -\frac{1}{2}|\grad \phi_\nu(\tilde{\bx}(t),t)|^2+\nu\triangle \phi_\nu(\tilde{\bx}(t),t). $$
Also, the forward stochastic flow $\tilde{\bx}_\nu(t)$ converges to the coalescing flow $\bx_*(t)$ for Burgers as 
$\nu\rightarrow 0$. Note that the laplacian term has been shown \cite{KhaninSobolevski10, KhaninSobolevski12} 
to have the limit along the trajectories of the forward coalescing flow given by  
$$ \lim_{\tau\downarrow 0}\lim_{\nu\downarrow 0} \nu\triangle \phi_\nu(\bx_*(t+\tau),t+\tau) = 
- \min_{\pm} D_L^{\bu_\pm(t)}\left(\bu_*(t),\bu_\pm(t)\right), $$
the Bregman divergence for the free-particle Lagrangian $L(t,\bx,\bv)=\frac{1}{2}|\bv|^2,$ or just the kinetic energy.
(Note the sign error in \cite{KhaninSobolevski10}, p.1591) The quantity $\partial_t\phi_\nu(\tilde{\bx}(t),t)$ 
then has an enticing similarity to our expression (\ref{lagdiss2}) for the  dissipative anomaly, when $\psi=L$. 
Unfortunately, we are skeptical that any general connection exists. A counterexample\footnote{There is also a physical 
puzzle what quantity would constitute the ``temperature'' to relate the ``entropy production'' $-k_B W[\tilde{\bx}]$
to energy dissipation.} is the stationary shock 
solution of viscous Burgers considered in section \ref{sec:scalar}, which has a kinetic energy anomaly 
$-\frac{2}{3}u_0^3\delta(x)$ in the limit $\nu\rightarrow 0,$ but for which $\partial_t\phi(x)=0!$ It remains to be
seen whether any ideas related to the fluctuation theorems can be at all connected with dissipative anomalies 
in Burgers or elsewhere.

\section{Final Discussion} 

Our work has verified that many of the relations suggested by the Kraichnan model \cite{Bernardetal98,GawedzkiVergassola00},
between Lagrangian particle stochasticity, anomalous dissipation, and turbulent weak solutions, remain valid for the
inviscid Burgers equation. Our results for Burgers give, as far as we are aware, the first proof of spontaneous stochasticity 
for a deterministic PDE problem. There is some similarity with the results of Brenier \cite{Brenier89} on 
global-in-time existence of action minimizers for incompressible Euler fluids via  ``generalized flows". However, 
unlike Brenier's work which dealt with a two-time boundary-value problem, our stochastic representation (\ref{Burgrep})
is valid for solutions of the Cauchy problem, like the similar representations for weak solutions in the Kraichnan model.  
As in Brenier's work, however, and unlike in the Kraichnan model, we find that the stochastic Lagrangian flows for 
inviscid Burgers are generally non-unique (even for entropy solutions). An important question left open by our work 
is whether existence of suitable stochastic processes of Lagrangian trajectories, which are backward Markov and 
for which the velocity is a backward martingale, uniquely characterize the entropy solution of Burgers.   

The most important outstanding scientific issue is certainly the validity of similar results for more physically realistic hydrodynamic 
equations, such as the incompressible Navier-Stokes equation. It is an entirely open mathematical question whether standard 
weak solutions of incompressible Euler can be obtained by the zero-viscosity limit of incompressible Navier-Stokes 
solutions and whether these Euler solutions are characterized by a backward martingale property for the fluid circulations, 
as earlier conjectured by us \cite{Eyink06, Eyink07, Eyink10}. It is not even known whether the ``arrow of time'' specified by the
martingale property is the same as the arrow specified by dissipation of energy. That is to say, it is unknown whether 
weak Euler solutions (if any) satisfying the backward martingale property for circulations must have kinetic energies 
always decreasing in time. There is not even a formal physicists' argument that this is so, let alone a rigorous proof. 

The existence of non-vanishing energy dissipation in the limit of zero viscosity has been termed the 
``zeroth law of turbulence'' \cite{Frisch95}. Explaining such anomalous dissipation is indeed the zeroth-order problem 
for any theory of turbulence. While much is known about turbulent energy cascade in Eulerian representation 
 from a synthesis of experiment, simulations and theory, the Lagrangian aspects remain rather mysterious. 
G. I. Taylor's vortex-stretching picture \cite{TaylorGreen37,Taylor38} is still the most common and popularly 
taught Lagrangian view of turbulent dissipation (e.g. see Feynman's undergraduate lectures \cite{Feynmanetal64}, 
volume II, section 41-5).  Taylor's line-stretching mechanism is exemplified  by the Kazantsev-Kraichnan model 
of kinematic magnetic dynamo in its ``free decay regime'', but this example also shows that Taylor's mechanism becomes 
much more subtle in the presence of spontaneous stochasticity \cite{EyinkNeto10}. We believe that the 
possibility exists for fundamentally new Lagrangian perspectives on turbulent energy dissipation for 
Navier-Stokes and related equations. We hope that the current work may provide some useful hints in that direction.

\vspace{20pt}
\noindent {\bf Acknowledgements}

G. E. thanks K. Khanin and A. Sobolevski for conversations about their work at the 2012 Wolfgang Pauli Institute 
workshop ``Mathematics of Particles and Flows'' and he also thanks many of the participants of the 2013 Eilat Workshop 
``Turbulence \& Amorphous Materials'' for their comments on our talk announcing the results of this paper,
in particular J. Bec, G. Falkovich, A. Frishman, I. Kolokolov, J. Kurchan, and K. Gaw\c{e}dzki. 
T.D. is grateful to Y.-K. Shi for some useful conversations.  Both authors thank the anonymous 
referees for their very helpful comments.


\setcounter{section}{0}
\renewcommand\thesection{\Alph{section}}

\section*{Appendix A: Multi-Shock Geometric Construction}\lb{multishock} 

The discussion in section \ref{sec:geom} assumed a single shock, but the main results (including Proposition 3.1 and 3.2) hold 
also with multiple shocks and mergers. We now discuss the construction of the random process and the verification 
of its properties in the general case. 

\begin{figure}[!ht]
\begin{center}
\includegraphics[height=3in,width=3.5in]{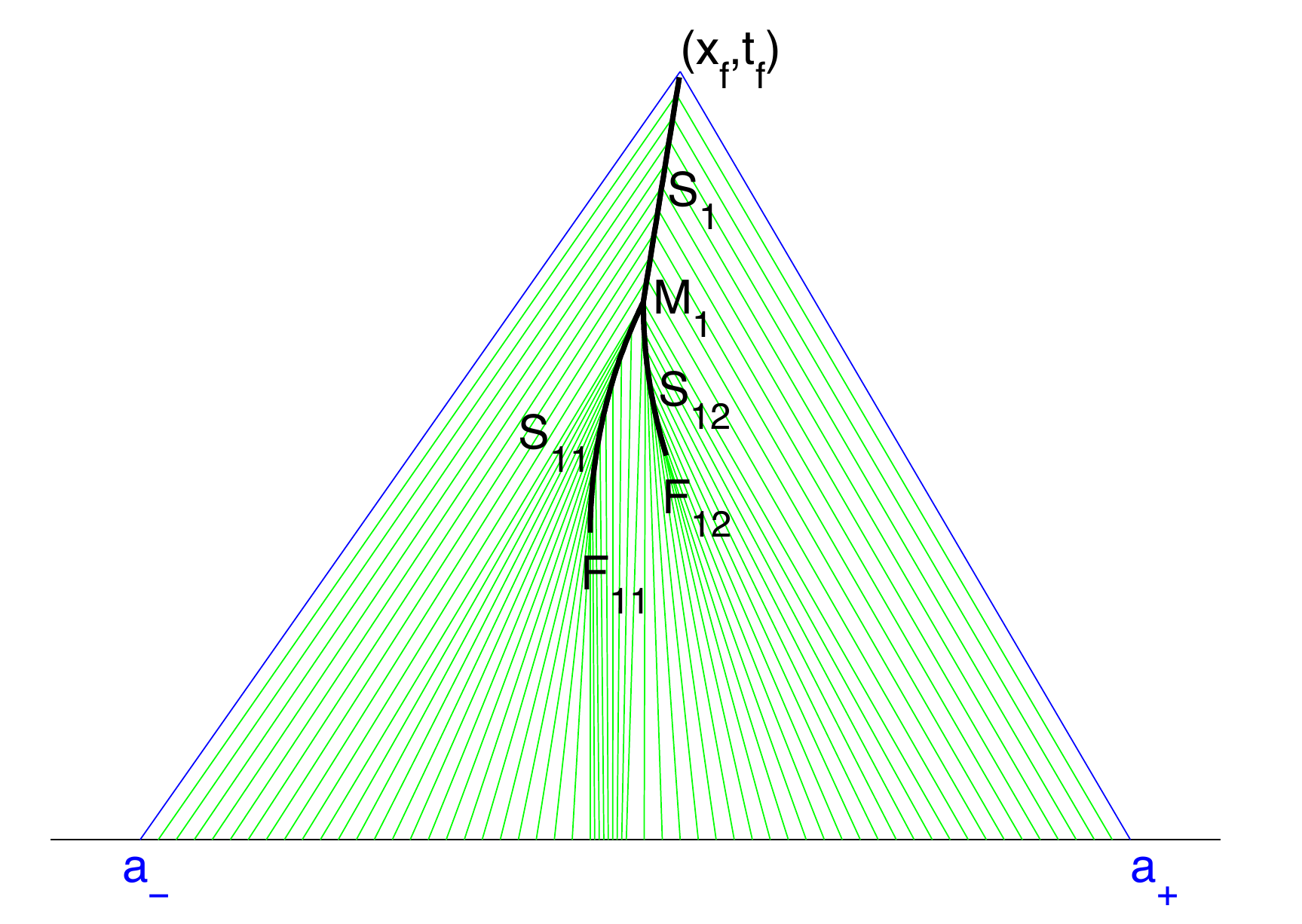}
\caption{{\small {\it Merger of Two Primitive Shocks.} Two shocks $\cS_{11}$ and $\cS_{12}$ 
form at the space-time points $\cF_{11}$ and $\cF_{12},$ then merge at point $\cM_1$ into the single 
shock $\cS_1.$ The straight green lines are some typical characteristic curves originating in the shock 
interval $[a_-,a_+]$ at time 0.} }
\end{center}
\end{figure}\lb{fig4}

We begin with the simplest example of a point $(x_f^*,t_f)$ located on a shock which resulted from the merger of two 
earlier ``primitive" shocks. The situation is illustrated by the space-time diagram in Fig.~4. The shock set $\cS$
in space-time consists of three segment curves. Two segments $\cS_{11},$ $\cS_{12}$ consist of shocks which
formed at points $\cF_{11}=(x_{F_{11}}^*,t_{F_{11}})$ and $\cF_{12}=(x_{F_{12}}^*,t_{F_{12}}),$ and then merged at point 
$\cM_1=(x_{M_{1}}^*,t_{M_{1}}).$ The third segment is the single shock $\cS_1$ which resulted from the merger, 
ending in the chosen point $(x_f^*,t_f).$ The random process is defined in the same manner as for the single shock case, 
by assigning to all shock segments the probability densities $p_1^\pm(\tau),$ $p_{1a}^\pm(\tau),\ a=1,2$ to leave 
the shock surface (backward in time) either to the right or left at time $\tau.$ These probabilities are assigned in the 
same way as before, by choosing a time $t_0<t_{F_{11}}$ (assuming here that $t_{F_{11}}\leq t_{F_{12}}$) and 
then mapping the uniform distribution on the shock interval $[b_-^f,b_+^f]$ at time $t_0$ into the shock set via 
the coalescing forward flow. This is illustrated in Fig.~4 for the case $(b,t_0)=(a,0),$ where the uniform distribution 
on the interval $[a_-^f,a_+^f]$ is mapped by the straight-line characteristics (green) to the shock. Note that the 
formula (\ref{ppm}) previously derived for the probability density still holds separately for each segment of the shock set. 

This assignment of probabilities again has the properties stated in the propositions of section \ref{sec:geom}.
Here we check the result ${\mathbb E}(\dot{x}(t))=u_f^*$ of Proposition 3.1 (which is also basic to Proposition 
3.2). If $t<t_{F_{12}},$ then we have the single shock case of section \ref{sec:geom}. However, if $t>t_{F_{12}},$
then one must take into account both shocks. First consider the case $t_{F_{12}}<t<t_{M_1},$ which is illustrated
in Fig.~5 for the case $(b,t_0)=(a,0).$ The probability distribution on the particle labels $[c_-^f,c_-^f]$ at time $t$ 
has a continuous part $p(c,t)$ and two atoms located at $x_{11}^*(t),x_{12}^*(t)$ on the shocks $\cS_{11},$ $\cS_{12}.$  
The two atoms have probabilities $\frac{b_{1+}(t)-b_{1-}(t)}{b_+^f-b_-^f}$ and $\frac{b_{2+}(t)-b_{2-}(t)}{b_+^f-b_-^f}$  
corresponding to the relative lengths of the intervals which map into those points. (See the magenta curves in Fig.~5). 
The continuous part of the distribution makes a contribution to ${\mathbb E}(\dot{x}(t))$ of the form 
$$ \int_{c_-^f}^{c_+^f} \rmd c \ u(c,t) \ p(c,t) = \frac{1}{b_+^f-b_-^f}\int_{[b_-^f,b_+^f]\backslash([b_{1-}(t),b_{1+}(t)]\cup [b_{2-}(t),b_{2+}(t)])} \rmd b \ u(b,t). $$
The contribution of each atom to ${\mathbb E}(\dot{x}(t))$ is
$$  \left(\frac{b_{a+}(t)-b_{a-}(t)}{b_+^f-b_-^f}\right)\cdot u_{1a}^*(t) = 
       \frac{1}{b_+^f-b_-^f}\int_{[b_{a-}(t),b_{a+}(t)]} \rmd b \ u(b,t), \,\,\,\, a=1,2, $$
using again the fundamental property (\ref{avrgcond}). 
Adding all of the contributions gives
$$ {\mathbb E}(\dot{x}(t)) =  \frac{1}{b_+^f-b_-^f} \int_{b_-^f}^{b_+^f} \rmd b \ u(b,t) = u_*^f. $$

\begin{figure}[!ht]
\begin{center}
\includegraphics[height=3in,width=3.5in]{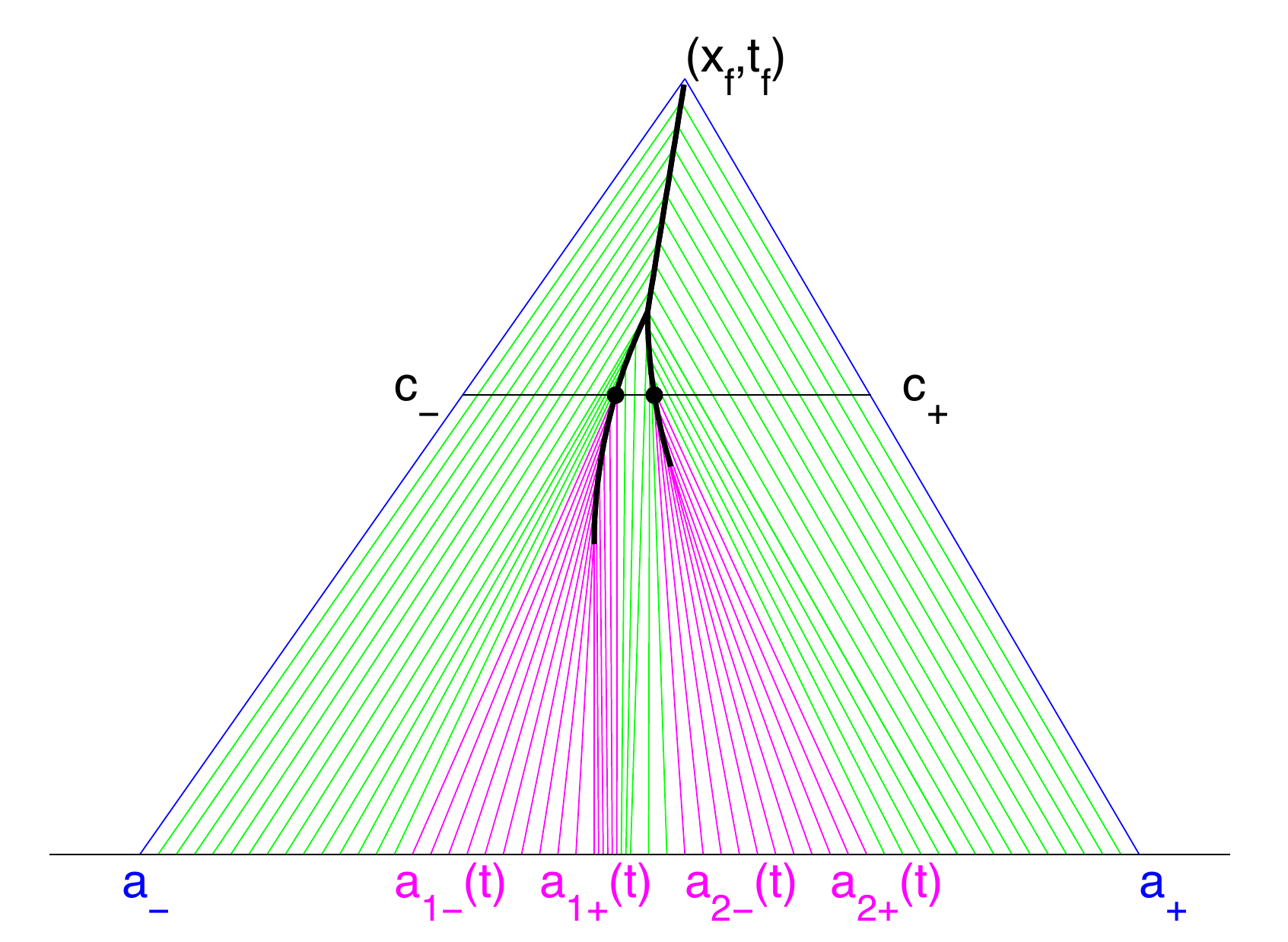}
\caption{{\small {\it Particle Positions with Two Shocks Before Merger}}. The black segment denotes the interval 
$[c_-^f,c_+^f]$ of particle positions at a time $t_{F_{12}}<t<t_{M_1}$. There are two atoms of finite probability 
located at $x_{11}^*(t),x_{12}^*(t)$ on the shocks $\cS_{11},$ $\cS_{12}.$ }
\end{center}
\end{figure}

\begin{figure}[!ht]
\begin{center}
\includegraphics[height=3in,width=3.5in]{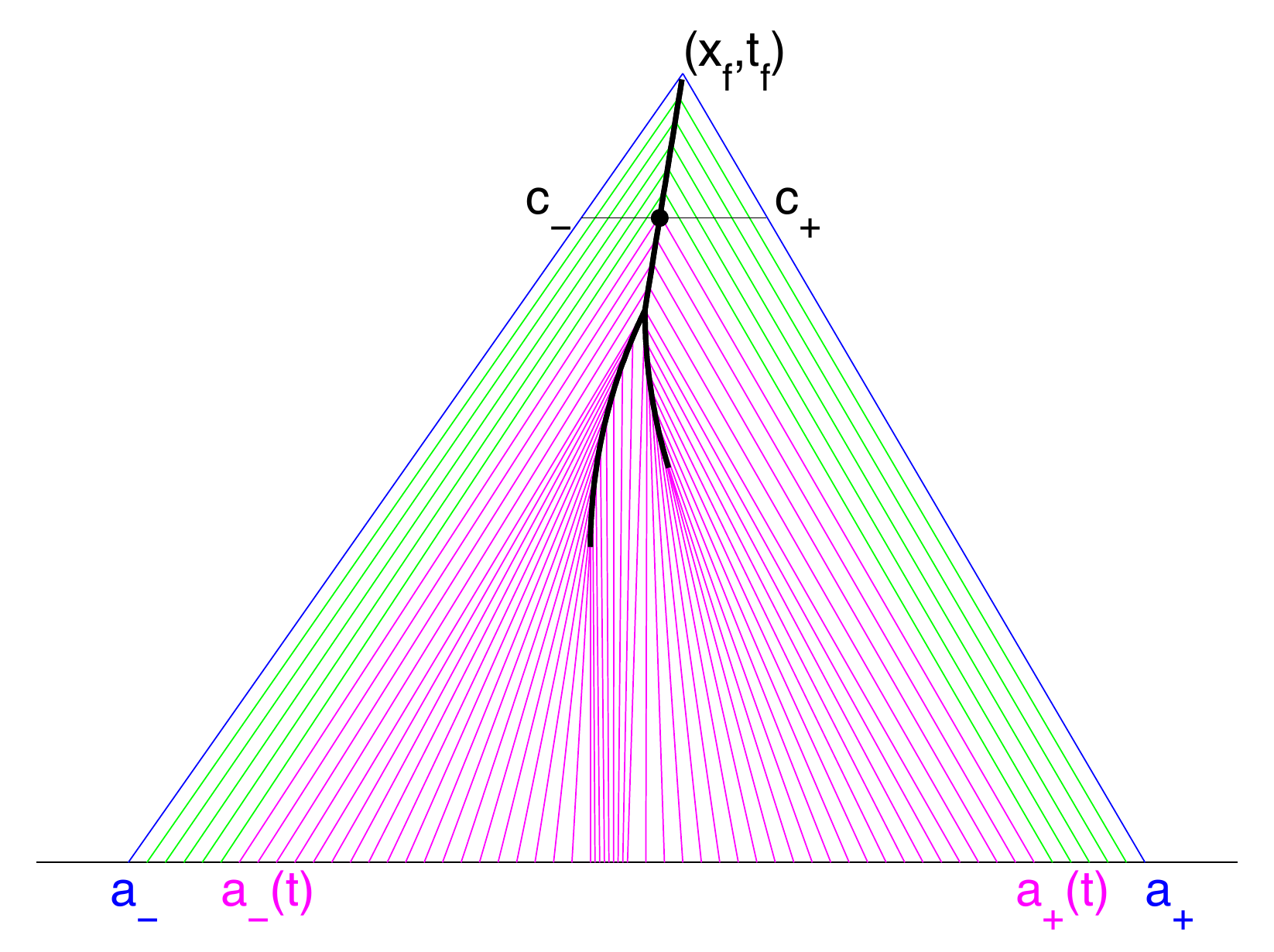}
\caption{{\small {\it Particle Positions with Two Shocks After Merger}}. The black segment denotes the interval 
$[c_-^f,c_+^f]$ of particle positions at a time $t_{M_1}<t<t_f$. There is a single atom of finite probability 
located at $x_{1}^*(t)$ on the shock $\cS_{1}.$}
\end{center}
\end{figure}

When instead $t_{M_1}<t<t_f,$ the situation is very much like the single shock case considered before (with any merger).
This situation is illustrated in Fig.~6 for the case $(b,t_0)=(a,0).$ The probability distribution on the particle labels $[c_-^f,c_-^f]$ 
at time $t$ has a continuous part $p(c,t)$ and one atom located at $x_{1}^*(t)$ on the shocks $\cS_{1}.$ It is readily seen 
by the same calculation as in section \ref{sec:geom} that ${\mathbb E}(\dot{x}(t))=u_f^*$. 

Although the definition of the random process and the verification of Propositions 3.1-2 are quite straightforward 
geometrically, the demonstration of its Markov properties backward in time become a bit cumbersome. To define 
the state-space of the (non-stationary) Markov process we must decompose the space-time set which lies 
outside the shock $\cS=\cS_1\cup \cS_{11}\cup \cS_{12}$ but which ends up at point $(x_f^*,t_f).$ We define
the right/left flanks of each shock segment $\cS_1^\pm,\cS_{1a}^\pm,$ $a=1,2$ as those points which are 
carried into the shock from the right/left by the forward coalescing flow. Set $\cS^\pm=\cS_1^\pm\cup \cS_{11}^\pm\cup \cS_{12}^\pm$
We also denote the straight-line characteristic  entering the merger point as $\cM_1^\downarrow$ and those entering the shock 
formation points as $\cF_{1a}^\downarrow,$ $a=1,2.$ The state space of the backward Markov process at time $\tau$
is then 
$$ X(\tau) = (\cS(\tau)\times \{-1,0,+1\}) \cup (\cM(\tau)\times \{0,\pm 1,\downarrow\} ) $$
$$ \dots \cup (\cS^+(\tau)\times \{+1\}) \cup (\cS^-(\tau)\times\{-1\})\cup (\cF^\downarrow(\tau)\times\{0\}) 
\cup (\cM^\downarrow(\tau)\times\{\downarrow\})$$
where we use the notation ${\mathcal A}(\tau)$ to denote the time-$\tau$ section of a space-time set ${\mathcal A}$.
These sections may be the empty set for some $\tau.$ The discrete labels indicate the property to be on 
\footnote{Note that on the curves $\cF_{1a}^\downarrow$ which will form shocks at times $t_{F_{1a}},$ $a=1,2,$
the label is also ``0.''} the shock $(0),$ right of the shock $(+)$, left of the shock $(-),$ and downward from the shock $(\downarrow).$  
The time-dependent infinitesimal generator $L(\tau)$ of the process is then obtained by straightforward calculations to be 
\begin{eqnarray}
&& t_M<\tau<t_f:\cr
&&\cr
&& \,\,\,\,\,\,\,\,\, L(\tau)f(x,\pm 1) = -u(x,\tau)f'(x,\pm 1), \ x\in \cS_{1}^\pm(\tau) \cr
&& \,\,\,\,\,\,\,\,\, L(\tau)f(x,0) = -\bar{u}(x,\tau)f'(x,0) + \sum_{\alpha=\pm 1} \lambda^\alpha_1(\tau) [f(x,\alpha)-f(x,0)], \ x\in \cS_{1}(\tau) \cr
&& \,\,\,\,\,\,\,\,\, L(\tau)f(x,\pm 1) = -u(x\pm,\tau)f'(x,\pm 1), \ x\in \cS_{1}(\tau) 
\end{eqnarray} 
with 
$$ \lambda^\pm_1(\tau) = p_1^\pm(\tau)/P_1(\tau), \,\,\,\, P_1(\tau)=1-\int_{\tau}^{t_f} \rmd t \ [p_1^+(t)+p_1^-(t)]. $$
\begin{eqnarray}
&& \tau=t_{M_1} \cr
&&\cr
&& \,\,\,\,\,\,\,\,\, L(\tau)f(x,\pm 1) = -u(x,t_M)f'(x,\pm 1), \ x\in \cS_{1}^\pm(t_M) \cr
&& \,\,\,\,\,\,\,\,\, L(\tau)f(x_M,0) = -\sum_{a=1,2} B_{1a} u_{1a}^*(t_{M_1})f'(x_M,0) 
+ \sum_{\alpha\in\{-1,\downarrow,+\}} \lambda^\alpha_M  [f(x_M,\alpha)-f(x_M,0)] \cr
&& \,\,\,\,\,\,\,\,\, L(\tau)f(x_M,\pm 1) = -u(x_M\pm,t_M)f'(x,\pm 1) \cr
&& \,\,\,\,\,\,\,\,\, L(\tau)f(x_M,\downarrow) = -u(x_M,t_M-)f'(x,\downarrow) 
\end{eqnarray} 
with 
$$ B_{1a} = \frac{P_{1a}}{P_{11}+P_{12}}, \,\,\,\,
P_{1a}=\int_{t_{F_{1a}}}^{t_{M_1}} \rmd t \ [p_{1a}^+(t) + p_{1a}^-(t)], \,\,\,\, a=1,2. $$
$$ \lambda^+_M= p_{12}^+(t_M)/(P_{11}+P_{12}), \,\,\,\,
\lambda^-_M= p_{11}^-(t_M)/(P_{11}+P_{12}), $$
$$ \lambda^\downarrow_M= (p_{11}^+(t_M) +p_{12}^-(t_M))/(P_{11}+P_{12}), $$
\begin{eqnarray}
&& t_{F_{12}}<\tau<t_M: \cr
&& \cr
&& \,\,\,\,\,\,\,\,\, L(\tau)f(x,\pm 1) = -u(x,\tau)f'(x,\pm 1), \ x\in S^\pm(\tau) \cr
&& \,\,\,\,\,\,\,\,\, L(\tau)f(x,0) = -\bar{u}(x,\tau)f'(x,0) + \sum_{\alpha=\pm 1} \lambda^\alpha_{1a}(\tau) [f(x,\alpha)-f(x,0)], \ x\in \cS_{1a}(\tau),
\,\,\,\, a=1,2 \cr
&& \,\,\,\,\,\,\,\,\, L(\tau)f(x,\pm 1) = -u(x\pm,\tau)f'(x,\pm 1), \ x\in \cS_{1a}(\tau),  \,\,\,\, a=1,2 \cr
&& \,\,\,\,\,\,\,\,\, L(\tau)f(x,\downarrow) = -u(x,\tau)f'(x,\downarrow), \ x\in M_{1}^\downarrow(\tau) 
\end{eqnarray} 
with 
$$ \lambda^\pm_{1a}(\tau) = p_{1a}^\pm(\tau)/P_{1a}(\tau), \,\,\,\, P_{1a}(\tau)=\int_{t_{F_{1a}}}^{\tau} \rmd t \ [p_{1a}^+(t)+p_{1a}^-(t)]. $$
\begin{eqnarray}
&& t_{F_{11}}<\tau<t_{F_{12}}: \cr
&& \cr 
&& \,\,\,\,\,\,\,\,\, L(\tau)f(x,\pm 1) = -u(x,\tau)f'(x,\pm 1), \ x\in S^\pm(\tau) \cr
&& \,\,\,\,\,\,\,\,\, L(\tau)f(x,0) = -\bar{u}(x,\tau)f'(x,0) + \sum_{\alpha=\pm 1} \lambda^\alpha_{11}(\tau) [f(x,\alpha)-f(x,0)], \ x\in \cS_{11}(\tau) \cr
&& \,\,\,\,\,\,\,\,\, L(\tau)f(x,\pm 1) = -u(x\pm,\tau)f'(x,\pm 1), \ x\in \cS_{11}(\tau) \cr
&& \,\,\,\,\,\,\,\,\, L(\tau)f(x,\downarrow) = -u(x,\tau)f'(x,\downarrow), \ x\in \cM_{1}^\downarrow(\tau) \cr
&& \,\,\,\,\,\,\,\,\, L(\tau)f(x,0) = -u(x,\tau)f'(x,0), \ x\in \cF_{12}^\downarrow(\tau) 
\end{eqnarray} 
\begin{eqnarray}
&& 0<\tau<t_{F_{11}}: \,\,\,\,\,\,\,\,\,\, L(\tau)f(x,\beta) = -u(x,\tau)f'(x,\beta), \ (x,\beta)\in X(\tau) 
\end{eqnarray} 
At $t=t_{M_1}$ the random particle on the shock segment $\cS_1$ moving backward in time branches to $\cS_{11}$  or $\cS_{12}$ 
with probabilities $B_{11},B_{12},$ respectively. The particle may also jump off the shock right, left or downward 
at the merger point $\cM_1.$ Note that 
$$ B_{11}u_{11}^*(t_{M_1})+B_{12}u_{12}^*(t_{M_1})=u_1^*(t_{M_1}) $$
as follows from the fundamental property (\ref{avrgcond}) applied to each of the shock segments $\cS_1, \cS_{11}, \cS_{12}$
at the merger point $\cM_1.$  

The above construction can be made completely general, using well-known facts about the entropy solution of 
inviscid Burgers \cite{BecKhanin07}. The fundamental property which enables our construction is that 
any straight-line characteristic which hits the shock set $\cS$ in space-time has not hit the shock set at any earlier time. 
The shock set $\cS$ at a fixed final time $t_f$ is, in general, a set of disconnected trees branching backward in time. 
The trees consists of segments that end (backward in time) either in a merger point or a formation point. At a merger 
point the shock segment branches into two or more shock segments. Generically, there will be exactly two shock segments 
branching off at each merger point, but non-generic mergers are possible that involve more than two lower segments 
branching off. At each formation point the shock set terminates backward in time. At the final time $t_f$ one thus starts with 
a countable number of shock segments $\cS_{n_1},$ $n_1=1,2,3,...$. Each shock segment $\cS_{n_1}$ ends in either a formation 
formation point $\cF_{n_1}$ or a merger point $\cM_{n_1}.$ At a merger point $\cM_{n_1}$ a second generation of shock segments appears,
typically $\cS_{n_11},$ $\cS_{n_12},$ but occasionally also $\cS_{n_13},$ $\cS_{n_14},$..., etc. These second-generation segments  
$\cS_{n_1n_2}$ also end either in a merger point $\cM_{n_1n_2}$ or a formation point $\cF_{n_1n_2}.$ Assuming that the inviscid Burgers 
solution started with smooth initial data $u_0$, this branching continues backward in time until every segment $\cS_{n_1n_2...n_p}$ ends 
in a formation point $\cF_{n_1n_2...n_p}.$ Call $\cT_{n_1}$ the tree of shock segments branching from $\cS_{n_1}$ (including
segment $\cS_{n_1}$ itself). Except for the merger and formation points, every point on a shock segment  $\cS_{n_1n_2...n_p}$ is intersected 
by exactly two straight-line characteristics curves, one from the left and one from the right. These characteristics for the tree $\cT_{n_1},$ 
$n_1=1,2,3,...$ all originate in a corresponding Lagrangian shock interval $[b_-^{(n_1)},b_+^{(n_1)}]$, $n_1=1,2,3,...$ at any time $t_0$
before the first shock has appeared.  A uniform distribution distribution on that interval then maps to an assignment of probabilities
$p_{\cS_{n_1n_2...n_p}}^\pm(\tau)$ for all segments $\cS_{n_1n_2...n_p}\in \cT_{n_1}.$ For each tree $\cT_{n_1}$, $n_1=1,2,3,...$
one has that 
$$ \sum_{\cS_{\#}\in \cT_{n_1}} \int_{t^i_\#}^{t^f_\#} \rmd \tau\ [p_{\cS_{\#}}^+(\tau)+p_{\cS_{\#}}^-(\tau)]=1, $$
where $t^i_\#,t^f_\#$ are the initial and final times of the segment $\cS_{\#}$. This specification of probabilities to leave the 
shock right or left backward in time fully specifies the random process $\tilde{x}(\tau)$ and it is easy to verify that it satisfies 
Propositions 3.1 and 3.2 by straightforward extensions of the previous arguments.  

\section*{Appendix B: Spontaneous Stochasticity in the Khokhlov Solution For $Pr<\infty$ }\lb{khokhlov} 

We present here the proof of spontaneous stochasticity of backward Lagrangian trajectories at any value 
of Prandtl number $Pr<\infty$ for the Khokhlov solution (\ref{khokhsol}) of viscous Burgers. This is a 
more generic example of a decaying Burgers solution than the stationary shock solution considered 
in section \ref{sec:scalar}. Although it blows up at time $t=0$ and is ill-defined there, it tends to zero 
in a typical way in the limit $t\rightarrow +\infty.$  

To describe evolution backward in time from some chosen time $t_f>0,$ we 
introduce the variable $\tau=\ln(t_f/t),$ so that $t=t_f e^{-\tau}$ and $t\downarrow 0$ as $\tau\uparrow
+\infty.$ The equation for backward-in-time stochastic trajectories thus by the change of time becomes 
\be  \rmd\txi = -[\txi - L\tanh(L\txi e^\tau/2\nu t_f)] \rmd\tau + \sqrt{2\kappa t_f e^{-\tau}} \ \rmd\tilde{W}(\tau). \lb{txieq-B} \ee
With $u_\nu(x,\tau)\equiv  -[x - L\tanh(L x e^\tau/2\nu t_f)]$ we introduce a potential 
$$ \phi_\nu(x,\tau) \equiv  -2\nu t_f e^{-\tau} \ln \cosh(Lx e^\tau/2\nu t_f) + \frac{1}{2}x^2, $$
so that $u_\nu(x,\tau)=-\rmd\phi_\nu(x,\tau)/\rmd x.$ This sign is opposite to that used previously in the paper,
but is chosen here so that $\phi_\nu$ acts formally as a Lyapunov function in the zero-noise limit and also so that 
$\phi_\nu(x,\tau)$ as a function of $x$ has the form of a typical ``double-well'' potential, with a local maximum at $x=0$ and 
a pair of global minima near $x=\pm L.$  Note, in particular, that 
$$ \phi_*(x) = \lim_{\nu \rightarrow 0} \phi_\nu(x,\tau) = \lim_{\tau \rightarrow \infty} \phi_\nu(x,\tau) = -L|x| + \frac{1}{2}x^2, $$
and $\phi_*(x)\leq \phi_\nu(x,\tau)$ for all $\nu, \tau>0.$ The naive $\nu,\kappa\rightarrow 0$ limit of (\ref{txieq-B}) is 
 \be  \rmd\xi = -[\xi - L\ {\rm sign}(\xi)] \rmd\tau,  \ee 
 with two deterministic solutions 
 $$\xi_\pm(\tau) = \pm L (1-e^{-\tau}) $$
 for $\tau>0,$ both satisfying $\xi_\pm(0)=0$ and
 $$ \phi_*(\xi_\pm(\tau)) = -\frac{1}{2} L^2 (1-e^{-2\tau}). $$
 The statement of spontaneous stochasticity that we prove here is:
 
 \begin{proposition}
 For solution $\txi(\tau)$ of (\ref{txieq-B}) with $\txi(0)=0$ at any $\tau>0,$ for $\alpha=\min\{1,\frac{1}{2Pr}\},$ and for any $\epsilon\in (0,1),$
 $$ \lim_{\stackrel{\kappa\rightarrow 0}{Pr=\nu/\kappa \,\,\,\, {\rm fixed}}} P_{\nu,\kappa}\Big( \phi_*(\txi(\tau),\tau) \leq 
 -\frac{1}{2} \alpha (1-\epsilon) L^2 (1-e^{-2\tau}) \Big) =1. $$
 \end{proposition} 
 
 \begin{proof}
 An application of the It$\bar{{\rm o}}$ lemma gives 
  \begin{eqnarray} 
 \rmd\big(e^{2\tau} \phi_\nu(\txi(\tau),\tau)\big) &=& -L^2\left[\tanh^2\left(\frac{L\txi e^\tau}{2\nu t_f}\right)+\frac{1}{2 Pr}
 {\rm sech}^2\left(\frac{L\txi e^\tau}{2\nu t_f}\right) \right]e^{2\tau}\rmd \tau \cr
 && +2\nu t_f \left[ \left(\frac{L\txi e^\tau}{2\nu t_f}\right)\tanh\left(\frac{L\txi e^\tau}{2\nu t_f}\right)-\ln\cosh\left(\frac{L\txi e^\tau}{2\nu t_f}\right)\right] e^{\tau} \rmd\tau \cr
 && +\kappa t_f e^{\tau} \rmd \tau - \sqrt{2\kappa t_f e^{-\tau}}\ u_\nu(\txi,\tau) \ e^{2\tau} \rmd \tilde{W}(\tau)  
 \end{eqnarray} 
 Using $\tanh^2(z)+{\rm sech}^2(z)=1,$ the first term on the righthand side is $\leq -\alpha L^2 e^{2\tau} \rmd \tau.$ Integrating 
 from 0 to $\tau$ thus yields
 $$ \phi_*(\txi(\tau)) \leq \phi_\nu(\txi(\tau),\tau) \leq  -\frac{1}{2}\alpha L^2 (1-e^{-2\tau}) + \kappa t_f e^{-\tau}(1-e^{-\tau}) + \tilde{R}_{\nu,\kappa}^{(1)}(\tau) + \tilde{R}_{\nu,\kappa}^{(2)}(\tau) $$
 with 
 $$ \tilde{R}_{\nu,\kappa}^{(1)}(\tau)=2\nu t_f e^{-2\tau}\int_0^\tau \left[ \left(\frac{L\txi(\sigma) e^\sigma}{2\nu t_f}\right)\tanh\left(\frac{L\txi(\sigma) e^\sigma}{2\nu t_f}\right)
 -\ln\cosh\left(\frac{L\txi(\sigma) e^\sigma}{2\nu t_f}\right)\right] e^{\sigma} \rmd\sigma $$
 and
 $$ \tilde{R}_{\nu,\kappa}^{(2)}(\tau) =  \sqrt{2\kappa t_f} e^{-2\tau}\int_0^\tau [\txi(\sigma) - L\tanh(L\txi(\sigma) e^\sigma/2\nu t_f)]  \ e^{3\sigma/2} \rmd \tilde{W}(\sigma).  $$
 Using $0\leq z\tanh(z) -\ln\cosh(z)\leq \ln 2,$ it follows that 
 $$0\leq \tilde{R}_{\nu,\kappa}^{(1)}(\tau) \leq \nu t_f (\ln 2) e^{-\tau} (1-e^{-\tau}) $$
 and thus
 \begin{eqnarray*}
 && \phi_*(\txi(\tau)) > -\frac{1}{2} \alpha (1-\epsilon) L^2 (1-e^{-2\tau}) \cr
 && \hspace{60pt} \Longrightarrow 
 \tilde{R}_{\nu,\kappa}^{(2)}(\tau) > \frac{1}{2} \epsilon \alpha L^2 (1-e^{-2\tau}) - (\kappa +\nu (\ln 2))t_f e^{-\tau}(1-e^{-\tau}). 
 \end{eqnarray*} 
 Note that the lower bound in the latter inequality is positive for sufficiently small $\nu,\kappa.$ To bound the probability of 
 this event, we can use Chebyshev inequality. By the It$\bar{\rm{o}}$ isometry
 \begin{eqnarray*}
 {\mathbb E}_{\nu,\kappa}\left(|\tilde{R}_{\nu,\kappa}^{(2)}(\tau)|^2\right)
 &= & 2 \kappa t_f e^{-4\tau} \int_0^\tau {\mathbb E}_{\nu,\kappa}\left|\txi(\sigma)-L\tanh(L\txi(\sigma) e^\sigma/2\nu t_f)\right|^2 e^{3\sigma}\rmd\sigma \cr
 &\leq & 2 \kappa t_f e^{-4\tau} \int_0^\tau \left[ {\mathbb E}_{\nu,\kappa}|\txi(\sigma)|^2+L^2\right] e^{3\sigma}\rmd\sigma. 
 \end{eqnarray*} 
 To obtain a bound on ${\mathbb E}_{\nu,\kappa}|\txi(\tau)|^2$ we again use the It$\bar{{\rm o}}$ lemma, as 
 \begin{eqnarray*} 
 \rmd(\txi^2) &=& 2\txi \rmd\txi + 2\kappa t_f e^{-\tau} \rmd\tau \cr
 &=& \left[-2\txi^2+2\txi L\tanh(L\txi(\tau) e^\tau/2\nu t_f)+2\kappa t_f e^{-\tau}\right]\rmd \tau + \rmd\tilde{M}
 \end{eqnarray*}
where $\tilde{M}(\tau)$ is a martingale with ${\mathbb E}_{\nu,\kappa}(\tilde{M})=0.$ By a Young's inequality 
$$ 2xL\tanh(L x e^\tau/2\nu t_f) \leq x^2 + L^2\tanh^2(Lx e^\tau/2\nu t_f)\leq x^2 + L^2, $$ 
so that 
$$ \rmd{\mathbb E}_{\nu,\kappa}(\txi^2) \leq \left(-{\mathbb E}_{\nu,\kappa}(\txi^2) +L^2 + 2\kappa t_f e^{-\tau}\right)\rmd \tau. $$
This integrates to 
$$ {\mathbb E}_{\nu,\kappa}\left(\txi^2(\tau)\right)\leq L^2(1-e^{-\tau}) + 2\kappa t_f\tau e^{-\tau} \leq L^2 + 2\kappa t_f \tau e^{-\tau} $$
and then 
$$ {\mathbb E}_{\nu,\kappa}\left(|\tilde{R}_{\nu,\kappa}^{(2)}(\tau)|^2\right) \leq 2\kappa t_f e^{-\tau}
\left[ \frac{2}{3}L^2(1-e^{-3\tau}) + \kappa t_f e^{-\tau} \left(\tau-\frac{1}{2}(1-e^{-2\tau}\right)\right]. $$ 
Since this vanishes as $\kappa\rightarrow 0,$ the Chebyshev inequality completes the proof. \hfill $\Box$

 \end{proof} 

\bibliographystyle{spmpsci} 
\bibliography{StochDissBurgers}

\end{document}